\documentclass[]{imaiai}

\usepackage{amssymb}
\usepackage{multirow}
\usepackage{graphicx}
\usepackage[caption=false]{subfig} 
\usepackage{pifont}
\usepackage{navigator} 
\usepackage{tabularx} 

\newcommand{\cmark}{\ding{51}}%
\newcommand{\xmark}{\ding{55}}%
\DeclareMathOperator{\Tr}{Tr} 

\usepackage{amsmath} 
\usepackage{amsthm} 
\usepackage{url} 

\begin{document}

\title{Measuring Partial Balance in Signed Networks}

\shorttitle{Measuring Partial Balance in Signed Networks} 
\shortauthorlist{S. Aref and M.C. Wilson} 

\author{{
\sc Samin Aref}$^*$
{\sc and}
{\sc Mark C. Wilson}\\[2pt]
Department of Computer Science, University of Auckland\\
Auckland, Private Bag 92019, New Zealand\\
$^*${\email{sare618@aucklanduni.ac.nz}}
}

\maketitle

\begin{abstract}
{Is the enemy of an enemy necessarily a friend? If not, to what extent does this tend to hold? Such questions were formulated in terms of signed (social) networks and necessary and sufficient conditions for a network to be ``balanced" were obtained around 1960. Since then the idea that signed networks tend over time to become more balanced has been widely used in several application areas. However, investigation of this hypothesis has been complicated by the lack of a standard measure of partial balance, since complete balance is almost never achieved in practice. We formalize the concept of a measure of partial balance, discuss various measures, compare the measures on synthetic datasets, and investigate their axiomatic properties. The synthetic data involves Erd\H{o}s-R\'{e}nyi and specially structured random graphs. We show that some measures behave better than others in terms of axioms and ability to differentiate between graphs. We also use well-known datasets from the sociology and biology literature, such as Read's New Guinean tribes, gene regulatory networks related to two organisms, and a network involving senate bill co-sponsorship. Our results show that substantially different levels of partial balance is observed under cycle-based, eigenvalue-based, and frustration-based measures. We make some recommendations for measures to be used in future work.}
{structural analysis, signed networks, balance theory, axiom, frustration index, algebraic conflict}
\\
2000 Math Subject Classification: 05C22, 05C38, 91D30, 90B10
\footnotetext{The reference to this article should be made as follows: {\scshape Aref, S., and Wilson, M.~C.}
	\newblock Measuring partial balance in signed networks.
	\newblock {\em Journal of Complex Networks 6}, 4 (2018), 566--595.
	\newblock doi:10.1093/comnet/cnx044.}
\end{abstract}

\section{Introduction} \label{s:intro}

Transitivity of relationships has a pivotal role in analyzing social interactions. Is the enemy of an enemy a friend? What about the friend of an enemy or the enemy of a friend? Network science is a key instrument in the quantitative analysis of such questions. Researchers in the field are interested in knowing the extent of transitivity of ties and its impact on the global structure and dynamics in communities with positive and negative relationships. Whether the application involves international relationships among states, friendships and enmities between people, or ties of trust and distrust formed among shareholders, relationship to a third entity tends to be influenced by immediate ties.

There is a growing body of literature that aims to connect theories of social structure with network science tools and techniques to study local behaviors and global structures in signed graphs that come up naturally in many unrelated areas. The building block of structural balance is a work by Heider \cite{heider_social_1944} that was expanded into a set of graph-theoretic concepts by Cartwright and Harary \cite{cartwright_structural_1956} to handle a social psychology problem a decade later. The relationship under study has an antonym or dual to be expressed by the opposite sign \cite{harary_structural_1957}. In a setting where the opposite of a negative relationship is a positive relationship, a tie to a distant neighbor can be expressed by the product of signs reaching him. Cycles containing an odd number of negative edges are considered to be unbalanced, guaranteeing total balance therefore only in networks containing no such cycles. This strict condition makes it quite unlikely for a signed network to be totally balanced. The literature on signed networks suggests many different formulae to measure balance. These measures are useful for detecting total balance and imbalance, but for intermediate cases their performance is not clear and has not been systematically studied.

\subsection*{\textbf{Our contribution}}
The main focus of this paper is to provide insight into measuring partial balance, as much uncertainty still exists on this. The dynamics leading to specific global structures in signed networks remain speculative even after studies with fine-grained approaches. The central thesis of this paper is that not all measures are equally useful. We provide a numerical comparison of several measures of partial balance on a variety of undirected signed networks, both randomly generated and inferred from well-known datasets. Using theoretical results for simple classes of graphs, we suggest an axiomatic framework to evaluate these measures and shed light on the methodological details involved in using such measures.

This paper begins by laying out the theoretical dimensions of the research in Section~\ref{s:problem} and looks at basic definitions and terminology. In Section~\ref{s:check} different means of checking for total balance are outlined. Section~\ref{s:measure} discusses some approaches to measuring partial balance in Eq. \eqref{eq1.5} -- \eqref{eq5.1}, categorized into three families of measures \ref{ss:familyc} -- \ref{ss:familyf} and summarized in Table~\ref{tab1}. Numerical results on synthetic data are provided in Figures \ref{fig1} -- \ref{fig2} in Section~\ref{s:basic}. Section~\ref{s:special} is concerned with analytical results on synthetic data in closed-form formulae in Table~\ref{tab2} and visually represented in Figures \ref{fig3} -- \ref{fig5}. Axioms and desirable properties are suggested in Section~\ref{s:axiom} to evaluate the measures systematically. Section~\ref{s:recom} concerns recommendations for choosing a measure of balance. Numerical results on real signed networks are presented in Section~\ref{s:real}. Finally, Section~\ref{s:conclu} summarizes the study and provides direction for future research. 

\section{Problem statement and notation} \label{s:problem}
Throughout this paper, the terms signed graph and signed network will be used interchangeably to refer to a graph with positive and negative edges. We use the term cycle only as a shorthand for referring to simple cycles of the graph. While several definitions of the concept of balance have been suggested, this paper will only use the definition for undirected signed graphs unless explicitly stated.

We consider an undirected signed network $G = (V,E,\sigma)$ where $V$ and $E$ are the sets of vertices and edges, and $\sigma$ is the sign function 
$\sigma: E\rightarrow\{-1,+1\}$. The set of nodes is denoted by $V$, with $|V| = n$. The set of edges is represented by $E$ including $m^-$ negative edges and $m^+$ positive edges adding up to a total of $m=m^+ + m^-$ edges. The symmetric \emph{signed adjacency matrix} and the \emph{unsigned adjacency matrix} are denoted by \textbf{A} and ${|\textbf{A}|}$ respectively. Their entries are defined in \eqref{eq1} and \eqref{eq1.1}. 
\begin{align}\label{eq1}
a{_u}{_v} =
\left\{
\begin{array}{ll}
\sigma_{(u,v)} & \mbox{if } {(u,v)}\in E \\
0 & \mbox{if } {(u,v)}\notin E
\end{array}
\right.
\end{align}

\begin{align}\label{eq1.1}
|a{_u}{_v}| =
\left\{
\begin{array}{ll}
1 & \mbox{if } {(u,v)}\in E \\
0 & \mbox{if } {(u,v)}\notin E
\end{array}
\right.
\end{align}

The \emph{positive degree} and \emph{negative degree} of node $i$ are denoted by $d^+ _{i}$ and $d^- _{i}$ representing the number of positive and negative edges incident on node $i$ respectively. They are calculated based on $d^+ _{i}=(\sum_j |a_{ij}| + \sum_j a_{ij})/2$ and $d^- _{i}=(\sum_j |a_{ij}| - \sum_j a_{ij})/2$. The \emph{degree} of node $i$ is represented by $d_i$ and equals the number of edges incident on node $i$. It is calculated based on $d_{i}= d^+ _{i} + d^- _{i}= \sum_j |a_{ij}|$. 

A \emph{walk} of length $k$ in $G$ is a sequence of nodes $v_0,v_1,...,v_{k-1},v_k$ such that for each $i=1,2,...,k$ there is an edge from $v_{i-1}$ to $v_i$. If $v_0=v_k$, the sequence is a \emph{closed walk} of length $k$. If all the nodes in a closed walk are distinct except the endpoints, it is a \emph{cycle} (simple cycle) of length $k$. The \emph{sign of a cycle} is the product of the signs of its edges. A cycle is \emph{balanced} if its sign is positive and is \emph{unbalanced} otherwise. The total number of balanced cycles (closed walks) is denoted by $O_k ^+$ ($Q_k ^+$). Similarly, $O_k ^-$ ($Q_k ^-$) denotes the total number of unbalanced cycles (closed walks). The total number of cycles (closed walks) is represented by $O_k=O_k ^+ + O_k ^-$ ($Q_k=Q_k ^+ + Q_k ^-$) .

\section{Checking for balance}\label{s:check}
It is essential to have an algorithmic means of checking for balance. We recall several known methods here. The characterization of \emph{bi-polarity} (also called bipartitionability), that a signed graph is balanced if and only if its vertex set can be partitioned into two subsets such that each negative edge joins vertices belonging to different subsets \cite{harary_notion_1953}, leads to an obvious breadth-first search procedure similar to the usual algorithm for determining whether a graph is bipartite. 
An alternate algebraic criterion is that the eigenvalues of the signed and unsigned adjacency matrices are equal if and only if the signed network is balanced \cite{acharya_spectral_1980}. For our purposes the following additional method of detecting balance is also important. We define the \emph{switching function} $g(X)$ operating over a set of vertices $X\subseteq V$ as follows.
\begin{align} \label{eq1.2}
\sigma^{g(X)} _{(u,v)}=
\left\{
\begin{array}{ll}
\sigma_{(u,v)} & \mbox{if } {u,v}\in X  \ \text{or} \ {u,v}\notin X  \\
-\sigma_{(u,v)} & \mbox{if } (u \in X \ \text{and} \ v\notin X) \ \text{or} \ (u \notin X \ \text{and} \ v \in X)
\end{array}
\right.
\end{align}
As the sign of cycles remains the same when $g$ is applied, any balanced graph can switch to an all-positive signature \cite{hansen_labelling_1978,harary_simple_1980}. Accordingly, a balance detection algorithm 
can be developed by constructing a switching rule on a spanning tree and a root vertex, as suggested in \cite{hansen_labelling_1978,harary_simple_1980}.
Finally, another method of checking for balance in connected signed networks makes use of the signed Laplacian matrix defined by $\textbf{L}=\textbf{D}-\textbf{A}$ where $\textbf{D}_{ii} = \sum_j |a_{ij}|$ is the diagonal matrix of degrees. The signed Laplacian matrix, $\textbf{L}$, is positive-semidefinite i.e. all of its eigenvalues are nonnegative \cite{zaslavsky1983signed,zaslavsky_matrices_2010}. The smallest eigenvalue of $\textbf{L}$ equals 0 if and only if the graph is balanced \cite[Section 8A]{zaslavsky1983signed}.

\section{Measures of partial balance} \label{s:measure}
Several ways of measuring the extent to which a graph is balanced have been introduced by researchers. We discuss three families of measures here and summarize them in Table~\ref{tab1}.

\subsection{Measures based on cycles}\label{ss:familyc}

The simplest of such measures is the \textit{degree of balance} suggested by Cartwright and Harary \cite{cartwright_structural_1956}, which is the fraction of balanced cycles:
\begin{align}\label{eq1.5}
D(G)= \frac {\sum \limits_{k=3}^n O_k ^+ } {\sum \limits_{k=3}^n O_k}
\end{align}

There are other cycle-based measures closely related to $D(G)$. The \emph{relative $k$-balance}, denoted by $D_k(G)$ and formulated in \eqref{eq1.6} is a cycle-based measure where the sums defining the numerator and denominator of $D(G)$ are restricted to a single term of fixed index $k$ \cite{harary_structural_1957,harary1977graphing}. The special case $k=3$ is called the \emph{triangle index}, denoted by $T(G)$.

\begin{align}\label{eq1.6}
D_k(G)= \frac {O_k ^+} {O_k}
\end{align}

Giscard et al.\ have recently introduced efficient algorithms for counting simple cycles \cite{giscard2016general} making it possible to use various measures related to $D_k(G)$ to evaluate balance in signed networks \cite{Giscard2016}. 

A generalization is \emph{weighted degree of balance}, obtained by weighting cycles based on length as in \eqref{eq1.7}, in which $f(k)$ is a monotonically decreasing nonnegative function of the length of the cycle. 

\begin{align}\label{eq1.7}
C(G)=\frac{\sum \limits_{k=3}^n f(k) O_k ^+ }{\sum \limits_{k=3}^n f(k) O_k}
\end{align}

The selection of an appropriate weighting function is briefly discussed by Norman and Roberts \cite{norman_derivation_1972}, suggesting functions such as $1/k,1/k^2,1/2^k$, but no objective criterion for choosing such a weighting function is known. We consider two weighting functions $1/k$ and $1/k!$ for evaluating $C(G)$ in this paper. Given the typical distribution of cycles of different length, the function $f(k)=1/k$ provides a rate of decay slower than the typical rate of increase in cycles of length $k$, resulting in $C(G)$ being dominated by longer cycles. The function $f(k)=1/k!$ provides a rate of decay faster than the typical rate of increase in cycles of length $k$, resulting in $C(G)$ being mostly determined by shorter cycles.

Although fast algorithms are developed for counting and listing cycles of undirected graphs \cite{birmele2013optimal, giscard2016general}, the number of cycles grows 
exponentially with the size of a typical real-world network. 
To tackle the computational complexity, Terzi and Winkler \cite{terzi_spectral_2011} used $D_3(G)$ in their study and replaced the triangles by closed walks of length $3$. The triangle index can be calculated efficiently by the formula in \eqref{eq1.8} where $\Tr(\textbf{A})$ denotes the trace of ${\textbf{A}}$. 
\begin{align}\label{eq1.8}
T(G)= D_3(G) = \frac{O_3 ^+}{O_3} = \frac{\Tr({\textbf{A}}^3)+\Tr({|\textbf{A}|}^3)}{2 \Tr({|\textbf{A}|}^3)} 
\end{align}

The \emph{relative signed clustering coefficient} is suggested as a measure of balance by Kunegis \cite{kunegis_applications_2014}, taking insight from the classic clustering coefficient. After normalization, this measure is equal to the triangle index. Having access to an easy-to-compute formula \cite{terzi_spectral_2011} for $T(G)$ obviates the need for a clustering-based calculation which requires iterating over all triads in the graph.

Bonacich argues that dissonance and tension are unclear in cycles of length greater than three \cite{bonacich_introduction_2012}, justifying the use of the triangle index to analyze structural balance. However, the neglected interactions may represent potential tension and dissonance, though not as strong as that represented by unbalanced triads
. One may consider a smaller weight for longer cycles, thereby reducing their impact rather than totally disregarding them. Note that $C(G)$ is a generalization of both $D(G)$ and $D_k(G)$. 

In all the cycle-based measures, we consider a value of $1$ for the case of division by zero. This allows the measures $D(G)$ and $C(G)$ ($D_{k}(G)$) to provide a value for acyclic graphs (graphs with no $k$-cycle).

\subsection{Measures related to eigenvalues}\label{ss:familye}

Beside checking cycles, there are computationally easier approaches to evaluating structural balance such as the walk-based approach. The \emph{walk-based measure of balance} is suggested by Pelino and Maimone \cite{pelino2012towards} with more weight placed on shorter closed walks than the longer ones. Let $\Tr(e^\textbf{A})$ and $\Tr(e^{|\textbf{A}|})$ denote the trace of the matrix exponential for $\textbf{A}$ and $|\textbf{A}|$ respectively. In Equation \eqref{eq2}, closed walks are weighted by a function with a relatively fast rate of decay compared to functions suggested in \cite{norman_derivation_1972}. The weighted ratio of balanced to total closed walks is formulated in \eqref{eq2}. 
\begin{align}\label{eq2}
W(G)= \frac {K(G)+1}{2}, \quad K(G)=\frac{\sum \limits_{k}\frac{Q_k ^+ - Q_k ^-}{k!}}{\sum \limits_{k} \frac{Q_k ^+ + Q_k ^-}{k!}}=\frac{\Tr(e^\textbf{A})}{\Tr(e^{|\textbf{A}|})}  
\end{align}
Regarding the calculation of $\Tr(e^\textbf{A})$, one may use the standard fact that $\textbf{A}$ is a symmetric matrix for undirected graphs. It follows that $\Tr(e^\textbf{A})=\sum _{i} e^{\lambda_i}$ in which $\lambda_i$ ranges over eigenvalues of $\textbf{A}$. The idea of a walk-based measure was then used by Estrada and Benzi \cite{estrada_walk-based_2014}. They have tested their measure on five signed networks resulting in values inclined towards imbalance which were in conflict with some previous observations \cite{facchetti_computing_2011,kunegis_applications_2014}. The walk-based measure of balance suggested in \cite{estrada_walk-based_2014} have been scrutinized in the subsequent studies \cite{Giscard2016,singh2017measuring}. Giscard et al. discuss how using closed-walks in which the edges might be repeated results in mixing the contribution of various cycle lengths and leads to values that are difficult to interpret \cite{Giscard2016}. Singh et al. criticize the walk-based measure from another perspective and explains how the inverse factorial weighting distorts the measure towards showing imbalance \cite{singh2017measuring}.

The idea of another eigenvalue-based measure comes from spectral graph theory \cite{kunegis_spectral_2010}. The smallest eigenvalue of the signed Laplacian matrix provides a measure of balance for connected graphs called \textit{algebraic conflict} \cite{kunegis_spectral_2010}. Algebraic conflict, denoted by $\lambda(G)$, equals zero if and only if the graph is balanced. Positive-semidefiniteness of $\textbf{L}$ results in $\lambda(G)$ representing the amount of imbalance in a signed network. Algebraic conflict is used in \cite{kunegis_applications_2014} to compare the level of balance in online signed networks of different sizes. Moreover, Pelino and Maimone analyzed signed network dynamics based on $\lambda(G)$ \cite{pelino2012towards}. Bounds for $\lambda(G)$ are investigated by \cite{Hou2004} leading to recent applicable results in \cite{Belardo2014,Belardo2016}. Belardo and Zhou prove that $\lambda(G)$ for a fixed $n$ is maximized by the complete all-negative graph of order $n$ \cite{Belardo2016}. Belardo shows that $\lambda(G)$ is bounded by $\lambda_\text{max}(G)=\overline{d}_{\text{max}} -1$ in which $\overline{d}_{\text{max}}$ represents the maximum average degree of endpoints over graph edges \cite{Belardo2014}. We use this upper bound to normalize algebraic conflict. \textit{Normalized algebraic conflict}, denoted by $A(G)$, is expressed in \eqref{eq3}. 
\begin{align}\label{eq3}
A(G)=1- \frac {\lambda(G)}{\overline{d}_{\text{max}} -1} , \quad \overline{d}_{\text{max}}=\max_{(u,v)\in E} (d_u+d_v)/2
\end{align}

\subsection{Measures based on frustration}\label{ss:familyf}

A quite different measure is the \emph{frustration index} \cite{abelson_symbolic_1958,harary_measurement_1959,zaslavsky_balanced_1987} that is also referred to as the \emph{line index for balance} \cite{harary_measurement_1959}. A set $E^*$ of edges is called \textit{minimum deletion set} if deleting all edges in $E^*$ results in a balanced graph and deleting edges from no smaller set leads to a balanced graph. The frustration index equals the cardinality of a minimum deletion set as in Eq.\ \eqref{eq5.0}. 
\begin{align}\label{eq5.0}
L(G)= {|E^*|}
\end{align}

Each edge in $E^*$ lies on an unbalanced cycle and every unbalanced cycle of the network contains an odd number of edges in $E^*$. Iacono et al. showed that $L(G)$ equals the minimum number of unbalanced fundamental cycles induced over all spanning trees of the graph \cite{iacono_determining_2010}. The graph resulted from deleting all edges in $E^*$ is called a \textit{balanced transformation} of a signed graph. 

Similarly, in a setting where each vertex is given a black or white color, if the endpoints of positive (negative) edges have different colors (same color), they are ``frustrated". The frustration index is therefore the smallest number of frustrated edges over all possible 2-colorings of the nodes.

$L(G)$ is hard to compute as the special case with all edges being negative is equivalent to the MAXCUT problem \cite{gary1979computers}, which is known to be NP-hard. There are upper bounds for the frustration index such as $L(G)\leq m^-$ which states the obvious result of removing all negative edges. 

Facchetti, Iacono, and Altafini have used computational methods related to Ising spin glass models to estimate the frustration index in relatively large online social networks \cite{facchetti_computing_2011}. Using an estimation of the frustration index obtained by a heuristic algorithm, they concluded that the online signed networks are extremely close to total balance; an observation that contradicts some other research studies like \cite{estrada_walk-based_2014}. 

The number of frustrated edges in special Erd\H{o}s-R\'{e}nyi graphs is analyzed by El Maftouhi, Manoussakis and Megalakaki \cite{el_maftouhi_balance_2012}. It follows a binomial distribution with parameters $n(n-1)/2$ and $p/2$ in which $p$ represents equal probabilities for positive and negative edges in Erd\H{o}s-R\'{e}nyi graph. Therefore, the expected number of frustrated edges is $n(n-1)p/4$. They also prove that such a network is almost always not balanced when $p \geq (\log2)/n$. It is straightforward to prove that frustration index is equal to the minimum number of negative edges over all switching functions \cite{zaslavsky_matrices_2010}. Moreover, if $m^- {(G^{g(X)})}=L(G)$ then every vertex under switching ${g(X)}$ satisfies $d^- _{v^{g(X)}} \leq d^+ _{v^{g(X)}}$. Petersdorf \cite{petersdorf_einige_1966} proves that the frustration index is bounded by $\lfloor \left(n-1 \right)^2/4 \rfloor$.

Bounds for the largest number of frustrated edges for a graph with $n$ nodes and $m$ edges are provided in \cite{akiyama_balancing_1981}. It follows that ${L(G)} \leq {m}/{2}$; an upper bound that is not necessarily tight.

Another upper bound for the frustration index is reported in \cite{iacono_determining_2010} referred to as the worst-case upper bound on the consistency deficit. However, the frustration index values in complete graphs with all negative edges shows that the upper bound is incorrect.

In order to compare with the other indices which take values in the unit interval and give the value $1$ for balanced graphs, we suggest \emph{normalized frustration index}, denoted by $F(G)$ and formulated in \eqref{eq5.1}. 

\begin{align}\label{eq5.1}
F(G)=1- \frac {L(G)}{m/2}
\end{align}

Using a different upper bound for normalizing the frustration index, we discuss another frustration-based measure in subsection \ref{ss:normal} and formulate it in Eq.\ \ref{eq5.8}.

\subsection{Other methods of evaluating balance}

Balance can also be analyzed by blockmodeling based on iteratively calculating Pearson moment correlations from the columns of $\textbf{A}$ \cite{doreian_generalized_2005}. Blockmodeling reveals increasingly homogeneous sets of vertices. 

Doreian and Mrvar discuss this approach in partitioning signed networks \cite{doreian_partitioning_2009}. Applying the method to Correlates of War data on positive and negative international relationships, they refute the hypothesis that signed networks gradually move towards balance using blockmodeling alongside some variations of $D(G)$ and $L(G)$ \cite{patrick_doreian_structural_2015}. 

Moreover, there are probabilistic methods that compare the expected number of balanced and unbalanced triangles in the signed network and its reshuffled version \cite{leskovec_signed_2010,yap_why_2015, Szell_multi, Szell_measure}. As long as these measures are used to evaluate balance, the result will not be different to what $T(G)$ provides alongside a basic statistical testing of its value against reshuffled networks.

Some researchers suggest that studying the structural dynamics of signed networks is more important than measuring balance \cite{cai2015particle,ma_memetic_2015}. This approach is usually associated with considering an energy function to be minimized by local graph operations decreasing the energy. However, the energy function is somehow a measure of network imbalance which requires a proper definition and investigation of axiomatic properties. Seven measures of partial balance investigated in this paper are outlined in Table~\ref{tab1}.

\begin{table}[ht]
	\centering
	\caption{Measures of partial balance summarized}
	\label{tab1}
	\begin{tabularx}{\textwidth}{ll}
 \hline
		Measure & \multicolumn{1}{c}{Name, Reference, and Description}                                                 \\ \hline
		$D(G)$    & \textit{Degree of balance} \cite{cartwright_structural_1956, harary_measurement_1959}                      \\
		& A cycle-based measure representing the ratio of balanced cycles                          \\
		$C(G)$    & \textit{Weighted degree of balance} \cite{norman_derivation_1972}                                 \\
		& An extension of $D(G)$ using cycles weighted by a non-increasing function of length  \\
		$D_k(G)$& \textit{Relative $k$-balance} \cite{harary_structural_1957, harary1977graphing} \\
		& An extension of $D(G)$ placing a non-zero weight only on cycles of length $k$  \\
		$T(G)$    &  \textit{Triangle index} \cite{terzi_spectral_2011, kunegis_applications_2014}                             \\
		& A triangle-based measure representing the ratio of balanced triangles ($D_3(G)$)          \\
		$W(G)$    & \textit{Walk-based measure of balance} \cite{pelino2012towards, estrada_walk-based_2014}          \\
		& A simplified extension of $D(G)$ replacing cycles by closed walks                        \\
		$A(G)$    & \textit{Normalized algebraic conflict} \cite{kunegis_spectral_2010, kunegis_applications_2014}      \\
		& A normalized measure using least eigenvalue of the Laplacian matrix          \\
		$F(G)$    & \textit{Normalized frustration index} \cite{harary_measurement_1959, facchetti_computing_2011}                    \\
		& Normalized minimum number of edges whose removal results in balance                                \\ \hline 
	\end{tabularx}
\end{table}

\subsection*{\textbf{Outline of the rest of the paper}}
We started by discussing balance in signed networks in Sections \ref{s:problem} and \ref{s:check} and reviewed different measures in Section~\ref{s:measure}. We will provide some observations on synthetic data in Figures \ref{fig1} -- \ref{fig5} in Sections \ref{s:basic} and \ref{s:special} to demonstrate the values of different measures. The reader who is not particularly interested in the analysis of measures using synthetic data may directly go to Section~\ref{s:axiom} in which we introduce axioms and desirable properties for measures of partial balance. In Section~\ref{s:recom}, we provide some recommendations on choosing a measure and discuss how using unjustified measures has led to conflicting observations in the literature. The numerical results on real signed networks are presented in Section~\ref{s:real}.

\section{Numerical results on synthetic data} \label{s:basic}
In this section, we start with a brief discussion on the relationship between negative edges and imbalance in networks. According to the definition of structural balance, all-positive signed graphs (merely containing positive edges) are totally balanced. Intuitively, one may expect that all-negative signed graphs are very unbalanced. Perhaps another intuition derived by assuming symmetry is that increasing the number of negative edges in a network reduces partial balance proportionally. We analyze partial balance in randomly generated graphs to evaluate these intuitions. Our motivation for analyzing balance in such graphs is to gain an understanding of the behaviour of the measures and their connections with signed graph parameters like $m^-,n$, and density. We use randomly generated graphs as the underlying structure is not important for this analysis. 

\subsection{Erd\H{o}s-R\'{e}nyi random network with various $m^-$}
\label{ss:erdos}
We calculate measures of partial balance, denoted by $\mu(G)$, for an Erd\H{o}s-R\'{e}nyi random network with a various number of negative edges. Figure~\ref{fig1} demonstrates the partial balance measured by different methods. For each data point, we report the average of 50 runs, each assigning negative edges at random to the fixed underlying graph. The subfigures (c) and (d) of Figure~\ref{fig1} show the mean along with $\pm 1$ standard deviation. 

\begin{figure}
	
	\subfloat[The mean values of seven measures of partial balance]{\includegraphics[height=2.4in]{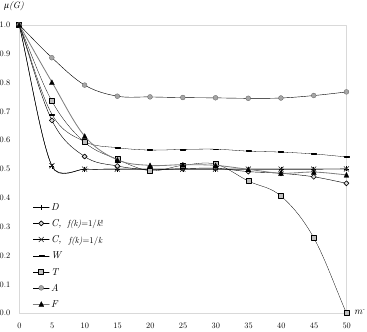}%
		\label{fig1_first_case}}
	\hfil
	\subfloat[The mean values of relative $k$-balance $D_k$]{\includegraphics[height=2.4in]{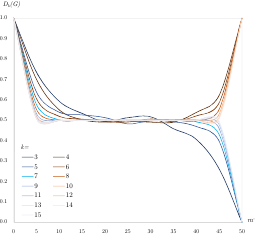}%
		\label{fig1_second_case}}
	
	~
	
	\subfloat[The standard deviation of $D$ and $C$]{\includegraphics[height=2.5in]{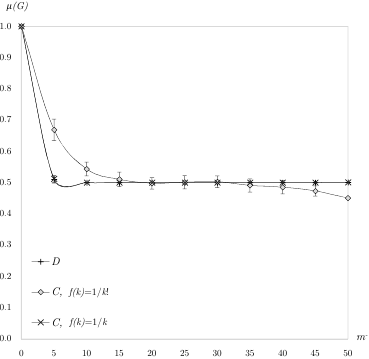}%
		\label{fig1_third_case}}
	\hfil
	\subfloat[The standard deviation of $W, T, A$ and $F$]{\includegraphics[height=2.5in]{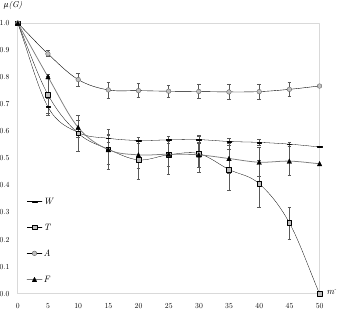}%
		\label{fig1_fourth_case}}
	
	\caption{Partial balance measured by different methods in Erd\H{o}s-R\'{e}nyi network with various number of negative edges (color version online)}
	\label{fig1}
\end{figure}

Measures $D(G)$ and $C(G)$ with $f(k)=1/k$ are observed to tend to $0.5$ where $m^- > 5$, not differentiating partial balance in graphs with a non-trivial number of negative edges. Given the typical distribution of cycles of different length, we expect $D(G)$ and $C(G)$ with $f(k)=1/k$ to be mostly determined by longer cycles that are much more frequent. For this particular graph, cycles with a length of 10 and above account for more than $96 \%$ of the total cycles in the graph. Such long cycles tend to be balanced roughly half the time for almost all values of $m^-$ (for all the values within the range of $5 \leq m^- \leq 45$ in the network considered here). The perfect overlap of data points for $D(G)$ and $C(G)$ with $f(k)=1/k$ in Figure~\ref{fig1} shows that using a linear rate of decay does not make a difference. One may think that if we use $D_k(G)$ which does not mix cycles of different length, it may circumvent the issues. However, subfigure (b) of Figure~\ref{fig1} demonstrating values of $D_k(G)$ for different cycle lengths shows the opposite. It shows not only does $D_k(G)$ not resolve the problems of lack of sensitivity and clustering around $0.5$, but it behaves unexpectedly with substantially different values based on the parity of $k$ when $m^- > 35$. $C(G)$ weighted by $f(k)=1/k!$ and mostly determined by shorter cycles, decreases slower than $D(G)$ and then provides values close to $0.5$ for $m^-\geq10$. $W(G)$ drops below $0.6$ for $m^-=10$ and then clusters around $0.55$ for $m^->10$. $T(G)$ is the measure with a wide range of values symmetric to $m^-$. The single most striking observation to emerge is that $A(G)$ seems to have a completely different range of values, which we discuss further in Subsection \ref{ss:normal}. A steady linear decrease is observed from $F(G)$ for $m^-\leq10$. 

\subsection{4-regular random networks of different orders}\label{ss:regular}
To investigate the impacts of graph order (number of nodes) and density on balance, we computed the measures for randomly generated 4-regular graphs with 50 percent negative edges. Intuitively we expect values to have low variation and no trends for similarly structured graphs of different orders. Figure~\ref{fig2} demonstrates the analysis in a setting where the degree of all the nodes remains constant, but the density ($4/n-1$) is decreasing in larger graphs. For each data point the average and standard deviation of 100 runs are reported. In each run, negative weights are randomly assigned to half of the edges in a fixed underlying 4-regular graph of order $n$.

\begin{figure}
	\centering
	\includegraphics[width=0.65\textwidth]{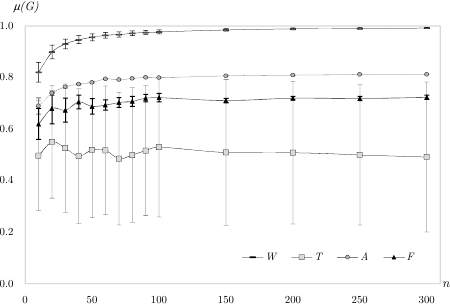}
	\caption{Partial balance measured by different methods in 50\% negative 4-regular graphs of different orders $n$ and decreasing densities $4/n-1$}
	\label{fig2}
\end{figure}

According to Figure~\ref{fig2}, the four measures differ not only in the range of values, but also in their sensitivity to the graph order and density. First, $W(G) \rightarrow 1$ when $n \rightarrow \infty$ for larger graphs although the graphs are structurally similar, which goes against intuition. Clustered around $0.5$ is $T(G)$ which features a substantial standard deviation for 4-regular random graphs. Values of $A(G)$ are around $0.8$ and do not seem to change substantially when $n$ increases. $F(G)$ provides stationary values around $0.7$ when $n$ increases. While $\lambda(G)$ and $L(G)$ depend on the graph order and size, the relative constancy of $A(G)$ and $F(G)$ values suggest the normalized measures $A(G)$ and $F(G)$ are largely independent of the graph size and order, as our intuition expects. We further discuss the normalization of $A(G)$ and $F(G)$ in Subsection \ref{ss:normal}.

\section{Analytical results on synthetic data} \label{s:special}
In this section, we analyze the capability of measuring partial balance in some families of specially structured graphs. Closed-form formulae for the measures in specially structured graphs are provided in Table~\ref{tab2}. The two families of complete signed graphs that we investigate are as follows in \ref{ss:kna} and \ref{ss:knc}. %

\subsection{Minimally unbalanced complete graphs with a single negative edge}\label{ss:kna}
The first family includes complete graphs with a single negative edge, denoted by ${K}_n^a$. Such graphs are only one edge away from a state of total balance. It is straight-forward to provide closed-form formulae for $\mu({K}_n^a)$ as expressed in \eqref{eq15} -- \eqref{eq18} in Appendix \ref{s:calc}.

\begin{table}[ht]
	\centering
	\caption{Balance in minimally and maximally unbalanced graphs ${K}_n^a$ (\ref{ss:kna}) and ${K}_n^c$ (\ref{ss:knc})}
	\label{tab2}
	\begin{tabular}{lll}
		\hline
		$\mu(G)$ & ${K}_n^a$                                                                              & ${K}_n^c$                                                                  \\ \hline
		$D(G)$      & $\sim 1- {2}/{n}$ & $\sim \frac{1}{2} + (-1)^n e^{-2}$     \\
		$C(G), f(k)=1/k!$      & $\sim 1- {1}/{n}$ &  $\sim \frac{1}{2} - \frac{3n\log n}{2^n}$\\
		$D_k(G)$    & $1-{2k}/{n(n-1)}$ & $0 , 1$                                       \\
		$W(G)$      & $\sim 1-{2}/{n}$  & $\sim \frac{1+e^{2-2n}}{2}$                 \\
		$A(G)$      & $\sim 1-{4}/{n^2}$& $0$                                        \\
		$F(G)$      & $1-{4}/{n(n-1)}$  & $\frac{1}{n},\frac{1}{n-1}$                \\ \hline
	\end{tabular}
\end{table}

\begin{figure}
	\centering
	\includegraphics[width=0.9\textwidth]{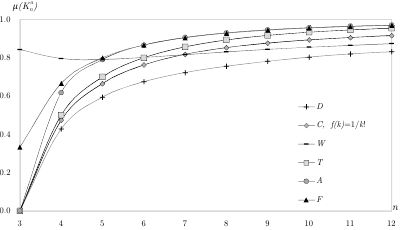}
	\includegraphics[width=0.9\textwidth]{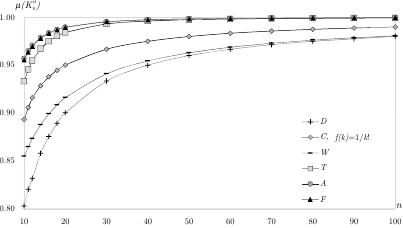}
	\caption{Partial balance measured by different methods for ${K}_n^a$ (\ref{ss:kna})}
	\label{fig3}
\end{figure}
In ${K}_n^a$, intuitively we expect $\mu({K}_n^a)$ to increase with $n$ and $\mu({K}_n^a) \rightarrow 1$ as $n \rightarrow \infty$. We also expect the measure to detect the imbalance in ${K}_3^a$ (a triangle with one negative edge). Figure~\ref{fig3} demonstrates the behavior of different indices for complete graphs with one negative edge. $W({K}_n^a)$ gives unreasonably large values for $n < 5$. Except for $W({K}_n^a)$, the measures are co-monotone over the given range of $n$.

\subsection{Maximally unbalanced complete graphs with all-negative edges}\label{ss:knc}
The second family of specially structured graphs to analyze includes all-negative complete graphs denoted by ${K}_n^c$. The indices are calculated in \eqref{eq21} -- \eqref{eq25} in Appendix \ref{s:calc}. 

Intuitively, we expect $\mu({K}_n^c) \rightarrow 0$ as $n \rightarrow \infty$. Figure~\ref{fig5} illustrates $D({K}_n^c)$ oscillating around $0.5$ and $W({K}_n^c),C({K}_n^c) \rightarrow 0.5$ as $n$ increases. We explain the oscillation of $D({K}_n^c)$ in Appendix \ref{s:calc}. Clearly, measures $D(G)$, $C(G)$ and, $W(G)$ do no work well for the family of graphs considered here. Figure~\ref{fig5} shows that $F({K}_n^c) \rightarrow 0$ as $n \rightarrow \infty$ as expected based on Table~\ref{tab2}.

\begin{figure}
	\centering
	\includegraphics[width=0.9\textwidth]{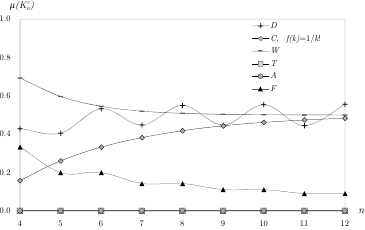}
	\includegraphics[width=0.9\textwidth]{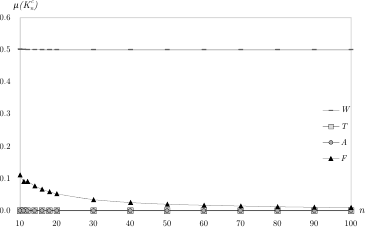}
	\caption{Partial balance measured by different methods for ${K}_n^c$ (\ref{ss:knc})}
	\label{fig5}
\end{figure}

\subsection{Normalization of the measures}\label{ss:normal}
It is worth mentioning that measures of partial balance may lead to different maximally unbalanced complete graphs. Based on $\lambda(G)$ and $L(G)$, ${K}_n^c$ represents the family of maximally unbalanced graphs, while it is merely one family among the maximally unbalanced graphs according to $T(G)$. Estrada and Benzi have found complete graphs comprised of one cycle of $n$ positive edges with the remaining pairs of nodes connected by negative edges to be a family of maximally unbalanced graphs based on $W(G)$ \cite{estrada_walk-based_2014}, while this argument is not supported by any other measures. As the signs of cycles in a graph are not independent, the structure of maximally unbalanced graphs under the cycle-based measures, $D(G)$ and $C(G)$, is not known. 

A simple comparison of $L({K}_n^c)$ (calculations provided in \eqref{eq24} in Appendix \ref{s:calc}) and the proposed upper bound $m/2=(n^2-n)/4$ reveals substantial gaps. These gaps equal $n/4$ for even $n$ and $(n-1)/4$ for odd $n$. This supports the previous discussions on looseness of $m/2$ as an upper bound for frustration index. As ${K}_n^c$ is maximally unbalanced under $L(G)$, $\lfloor m/2-(n-1)/4 \rfloor$ can be used as a tight upper bound for normalizing the frustration index. This allows a modified version of normalized frustration index, denoted by $F^\prime (G)$ and defined in \eqref{eq5.8}, to take the value zero for ${K}_n^c$.
\begin{align}
\label{eq5.8}
F^\prime (G)=1-L(G)/ \lfloor m/2-(n-1)/4 \rfloor
\end{align}

Similarly, the upper bound, $\lambda_\text{max}(G)$, used to normalize algebraic conflict, is not tight for many graphs. For instance, in the Erd\H{o}s-R\'{e}nyi studied in Section~\ref{s:basic} with $m=m^-$, the existence of an edge with $\overline{d}_{\text{max}}=9$ makes $\lambda_\text{max}(G)=8$, while $\lambda(G)=1.98$. 

The two observations mentioned above suggest that tighter upper bounds can be used for normalization. However, the statistical analysis we use in Section~\ref{s:real} to evaluate balance in real networks is independent of the normalization method, so we do not pursue this question further now.

\subsection{Expected values of the cycle-based measures}\label{ss:expected}

Relative $k$-balance, $D_k (G)$, is proved by El Maftouhi, Manoussakis and Megalakaki \cite{el_maftouhi_balance_2012} to tend to $0.5$ for Erd\H{o}s-R\'{e}nyi graphs such that the probability of an edge being negative is equal to $0.5$. Moreover, Giscard et al. discuss the probability distribution of  $1 - D_k(G)$. Their discussion is based on a model in which sign of any edge is negative with a fixed probability \cite[Section 4.2]{Giscard2016}. We use the same model to present some simple observations that appear not to have been noticed by previous authors advocating for the use of cycle-based measures. We are going to take a different approach from that of Giscard et al. and merely calculate the expected values of cycle-based measures in general, rather than the full distribution under additional assumptions. Note that, for an arbitrary graph, ${O_k ^+}/ {O_k}$ gives the probability that a randomly chosen $k$-cycle is balanced and is denoted by $B_{(k,q)}$.
\begin{theorem}
	\label{thm:random cycle}
	Let $G$ be a graph and consider the sign function obtained by independently choosing each edge to be negative with probability $q$, and positive otherwise. Then
	\begin{align}
	E(D_k(G)) & = (1+{(1-2q)}^k)/2 \label{eq1.9}
	\end{align}
\end{theorem}
\begin{proof}
	Note that a cycle is balanced if and only if it has an even number of negative edges. Thus $$E\left(\frac {O_k ^+} {O_k}\right) = \sum \limits_{i \: \text{even}} {\binom{k}{i}} q^i (1-q)^{k-i}$$ (compare with \cite[Eq. 4.1]{Giscard2016}). This simplifies to the stated formula (details of calculations are given in Appendix \ref{s:calc}).
\end{proof}

Note that the expected values are independent of the graph structure and obtaining them does not require making any assumptions on the signs of cycles being independent random variables. As the signs of the edges are independent random variables, the expected value of $B_{(k,q)}$ can be obtained by summing on all cases having an even number of negative signs in the $k$-cycle. 

Based on \eqref{eq1.9}, $E(D_k(G)) = 1$ when $q = 0$ and $E(D_k(G)) = 0.5$ when $q = 0.5$ supporting our intuitive expectations. However, when $q = 1$, $E(D_k(G))$ takes extremal values based on the parity of $k$ which is a major problem as previously observed in the subfigure (b) of Figure~\ref{fig1}. It is clear to see that the parity of $k$ makes a substantial difference to $D_k(G)$ when a considerable proportion of edges are negative.

\begin{theorem}
	\label{thm:randomallcycle}
	Let $G$ be a graph and consider the sign function obtained by independently choosing each edge to be negative with probability $q$, and positive otherwise. Then
	\begin{align}
	E(D(G)) & = \frac{1}{2} \frac { \sum \limits_{k=3}^n (1+{(1-2q)}^k)(O_k) } {\sum \limits_{k=3}^n O_k} \label{eq1.10}
	\end{align}
\end{theorem}

\begin{proof}
	The random variable ${O_k ^+}$ can be written as ${O_k ^+}= B_{(k,q)}.{O_k}$. Taking expected value from the two sides gives $E({O_k ^+})= {O_k}.E(B_{(k,q)})$ as $O_k$ is a constant for a fixed $k$. This completes the proof using the result from Theorem \ref{thm:random cycle}.
\end{proof}

Note that the exponential decay of the factor $(1-2q)^k$ reduces the contribution for large $k$, and small values of $k$ will dominate for many graphs. For example, if $q=0.2$ the expression for $E(D(G))$ simplifies to 
$$
\frac{1}{2} + \frac{1}{2} \frac{ \sum \limits_{k=3}^n{0.6}^k O_k}{\sum \limits_{k=3}^n O_k}.
$$
For many graphs encountered in practice, $O_k$ will grow exponentially with $k$, but at a rate less than $1/0.6$, so the tail contribution will be small. Larger values of $q$ only make this effect more pronounced. Thus we expect that $E(D(G))$ will often be very close to $0.5$ in signed graphs with a reasonably large fraction of negative edges (we have already seen such a phenomenon in Subsection~\ref{ss:erdos}). A similar conclusion can be made for $C(G)$. This casts doubt on the usefulness of the measures that mix cycles of different length whether weighted or not. 

While we have also observed many problems involving values of cycle-based measures on synthetic data in other parts of Sections \ref{s:basic} and \ref{s:special}, we will continue evaluating their axiomatic properties in Section~\ref{s:axiom} and then summarize the methodological findings in Section~\ref{s:recom}.

\section{Axiomatic framework of evaluation} \label{s:axiom}
The results in Section~\ref{s:basic} and Section~\ref{s:special} indicate that the choice of measure substantially affects the values of partial balance. Besides that, the lack of a standard measure calls for a framework of comparing different methods. Two different sets of axioms are suggested in \cite{norman_derivation_1972}, which characterize the measure $C(G)$ inside a smaller family (up to the choice of $f(k)$). Moreover, the theory of structural balance itself is axiomatized in \cite{schwartz_friend_2010}. However, to our knowledge, axioms for general measures of balance have never been developed. Here we provide the first set of axioms and desirable properties for measures of partial balance, in order to shed light on their characteristics and performance.

\subsection{Axioms for measures of partial balance}
We define a measure of partial balance to be a function $\mu$ taking each signed graph to an element of $[0,1]$. Worthy of mention is that some of these measures were originally defined as a measure of imbalance (algebraic conflict, frustration index and the original walk-based measure) calibrated at $0$ for completely balanced structures, so that some normalization was required, and perhaps our normalization choices can be improved on (see Subsection \ref{ss:normal}). As the choice of $m/2$ as the upper bound for normalizing the line index of balance was somewhat arbitrary, another normalized version of frustration index is defined in \eqref{eq6}. 
\begin{align}
\label{eq6}	
X(G)=1-L(G)/{m^-}
\end{align}

Before listing the axioms, we justify the need for an axiomatic evaluation of balance measures. As an attempt to understand the need for axiomatizing measures of balance, we introduce two unsophisticated and trivial measures that come to mind for measuring balance. The fraction of positive edges, denoted by $Y(G)$, is defined in \eqref{eq7} on the basis that all-positive signed graphs are balanced. Moreover, a binary measure of balance, denoted by $Z(G)$, is defined in \eqref{eq8}. While $Y(G)$ and $Z(G)$ appear to be irrelevant, there is currently no reason not to use such measures.
\begin{align}
\label{eq7}	
Y(G)=m^{+}/m \quad 
\end{align}
\begin{align}
\label{eq8}
Z(G) =
\left\{
\begin{array}{ll}
1 & \mbox{if } $G \:$ \mbox{is totally balanced} \\
0 & \mbox{if } $G \:$ \mbox{is not balanced}
\end{array} \right.
\end{align}

We consider the following notation for referring to basic operations on signed graphs:

\noindent
$G^{g(X)}$ denotes signed graph $G$ switched by $g(X)$ (switched graph).\\
$G\oplus H$ denotes the disjoint union of two signed graphs $G$ and $H$ (disjoint union).\\
$G\ominus e$ denotes $G$ with $e$ deleted (removing an edge).\\
$G\ominus E^*$ denotes $G$ with minimum deletion edges removed (balanced transformation).\\
$G\oplus C^+_3$ denotes the disjoint union of graphs $G$ and a positive 3-cycle (adding a balanced 3-cycle).\\
$G\oplus C^-_3$ denotes the disjoint union of graphs $G$ and a negative 3-cycle (adding an unbalanced 3-cycle).\\
$e\in E^*$ denotes an edge in the minimum deletion set.\\
$G \ominus E^* \oplus e$ denotes the balanced transformation of a graph with an edge $e$ added to it.\\

We list the following axioms:
\begin{description}
	\item[A1] $0 \leq \mu(G) \leq 1$. 
	\item[A2] $\mu(G) = 1$ if and only if $G$ is balanced. 
	\item[A3] If $\mu(G) \leq \mu(H)$, then $\mu(G) \leq \mu(G\oplus H) \leq \mu(H)$.	 
	\item[A4] $\mu(G^{g(X)}) = \mu(G)$.
\end{description}

The justifications for such axioms are connected to very basic concepts in balance theory. We consider A1 essential in order to make meaningful comparisons between measures. Introducing the notion of partial balance, we argue that total balance, being the extreme case of partial balance, should be denoted by an extremal value as in A2. In A3, the argument is that the overall balance of two disjoint graphs is bounded between their individual balances. This also covers the basic requirement that the disjoint union of two copies of graph $G$ must have the same value of partial balance as $G$. Switching nodes should not change balance \cite{zaslavsky_matrices_2010} as in A4.

Table \ref{tab3} shows how some measures fail on particular axioms. The results provide important insights into how some of the measures are not suitable for measuring partial balance. A more detailed discussion on the proof ideas and counterexamples related to Table~\ref{tab3} is provided in Appendix \ref{s:counter}.
\begin{table}[hbtp]
	\centering
	\caption{Different measures satisfying(\cmark) or failing(\xmark) axioms}
	\label{tab3}
		\begin{tabular}{p{0.1cm}p{0.1cm}p{0.1cm}p{0.1cm}p{0.1cm}p{0.1cm}p{0.1cm}p{0.1cm}p{0.1cm}l}
 \hline
			   & $D(G)$ & $C(G)$ & $W(G)$ & $D_k(G)$ & $A(G)$ & $F(G)$ & $X(G)$ & $Y(G)$ & $Z(G)$ \\ \hline
			A1 & \cmark & \cmark & \cmark & \cmark &  \cmark & \cmark & \cmark & \cmark & \cmark \\
			A2 & \cmark & \cmark & \cmark & \xmark &  \xmark & \cmark & \cmark & \xmark & \cmark \\
			A3 & \cmark & \cmark & \cmark & \cmark &  \xmark & \cmark & \cmark & \cmark & \cmark \\
			A4 & \cmark & \cmark & \cmark & \cmark &  \cmark & \cmark & \xmark & \xmark & \cmark \\
			 \hline
		\end{tabular}
\end{table}
\subsection{Some other desirable properties}

We also consider four desirable properties that formalize our expectations of a measure of partial balance. We do not consider the following as axioms in that they are based on adding or removing 3-cycles and edges which may bias the comparison in favor of cycle-based and frustration-based measures.

Positive and negative 3-cycles are very commonly used to explain the theory of structural balance which makes B1 and B2 obvious requirements. Removing an edge from a minimum deletion set, should not decrease balance as in B3. Finally, adding such an edge should not increase balance as in B4.

\begin{description}
		\item[B1]  If $\mu(G) \neq 1$, then $\mu(G\oplus C^+_3) > \mu(G)$. 
		\item[B2]  If $\mu(G) \neq 0$, then $\mu(G\oplus C^-_3) < \mu(G)$
		\item[B3]  If $e\in E^*$, then $\mu(G \ominus e) \geq \mu(G)$.
		\item[B4]  If $\mu(G)\neq 0$ and $\mu(G \ominus E^* \oplus e)\neq 1$, then $\mu(G \oplus e) \leq \mu(G)$.
\end{description}

Table \ref{tab3.5} shows how some measures fail on particular desirable properties. It is worth mentioning that the evaluation in Tables \ref{tab3}--\ref{tab3.5} is somewhat independent of parametrization: for each strictly increasing function $h$ such that $h(0)=0$ and $h(1)=1$, the results in Tables \ref{tab3}--\ref{tab3.5} hold for $h(\mu(G))$. Proof ideas and counterexamples related to Table~\ref{tab3.5} is provided in Appendix \ref{s:counter}.

\begin{table}[hbtp]
	\centering
	\caption{Different measures satisfying(\cmark) or failing(\xmark) desirable properties}
	\label{tab3.5}
	\begin{tabular}{p{0.1cm}p{0.1cm}p{0.1cm}p{0.1cm}p{0.1cm}p{0.1cm}p{0.1cm}p{0.1cm}p{0.1cm}l}
		\hline
		   & $D(G)$ & $C(G)$ & $W(G)$ & $D_k(G)$ & $A(G)$  & $F(G)$ & $X(G)$ & $Y(G)$ & $Z(G)$ \\ \hline
		B1 & \cmark & \cmark & \cmark & \xmark &  \cmark & \cmark & \xmark & \xmark & \xmark \\
		B2 & \cmark & \cmark & \xmark & \xmark &  \xmark & \xmark & \xmark & \xmark & \cmark \\
		B3 & \xmark & \xmark & \xmark & \xmark &  \xmark & \cmark & \cmark & \xmark & \xmark \\
		B4 & \xmark & \xmark & \xmark & \xmark &  \xmark & \cmark & \cmark & \xmark & \cmark \\ 
		\hline
	\end{tabular}
\end{table}

Another desirable property, which we have not formulated as a formal requirement owing to its vagueness, is that the measure takes on a wide range of values. For example, $D(G)$ and $C(G)$ tend rapidly to $0.5$ as $n$ increases which makes their interpretation and possibly comparison with other measures difficult. A possible way to formalize it would be expecting $\mu(G)$ to give $0$ and $1$ on each complete graph of order at least $3$, for some assignment of signs of edges. This condition would be satisfied by $T(G)$ and $A(G)$, as well as $F^\prime(G)$. However, $D(G),C(G)$ and $W(G)$ would not satisfy this condition due to the existence of balanced cycles and closed walks in complete signed graphs of general orders.
Moreover, the very small standard deviation of $D(G)$, $C(G)$, and $W(G)$ makes statistical testing against the balance of reshuffled networks complicated. The measures $D(G)$, $C(G)$, and $W(G)$ also have shown some unexpected behaviors on various types of graphs discussed in Section~\ref{s:basic} and Section~\ref{s:special}.

\section{Discussion on methodological findings} \label{s:recom}

Taken together, the findings in Sections \ref{s:basic} -- \ref{s:axiom} give strong reason not to use cycle-based measures $D(G)$ and $C(G)$, regardless of the weights. The major issues with cycle-based measures $D(G)$ and $C(G)$ include the very small variance in randomly generated and reshuffled graphs, lack of sensitivity and clustering of values around 0.5 for graphs with a non-trivial number of negative edges. Recall the numerical analysis of synthetic data in Section~\ref{s:basic}, analytical results on the expected values of cycle-based measures in Subsection \ref{ss:expected}, and the numerical values which are difficult to interpret like the oscillation of $D(G)$ and values of $C(G)$ for ${K}_n^c$ graphs in Table~\ref{tab2} and Figure~\ref{fig5}.

The relative $k$-balance which is ultimately from the same family of measures, seems to resolve some, but not all the problems discussed above. However, it fails on several axioms and desirable properties. It is easy to compute $D_{3}(G)$ based on closed walks of length 3 \cite{terzi_spectral_2011} and there are recent methods resolving the computational burden of computing $D_{k}(G)$ for general $k$ \cite{Giscard2016,giscard2016general}. However, $D_{k>3}(G)$ cannot be used for cyclic graphs that do not have $k$-cycles. Besides, for networks with a large proportion of negative edges, the parity of $k$ substantially distorts the values of $D_{k}(G)$. Accepting all these shortcomings, one may use $D_{k}(G)$ when cycles of a particular length have a meaningful interpretation in the context of study.

Walk-based measures like $W(G)$ require a more systematic way of weighting to correct for the double-counting of closed walks with repeated edges. The shortcomings of $W(G)$ involving the weighting method and contribution of non-simple cycles are also discussed in \cite{Giscard2016,singh2017measuring}. Recall that $W(G) \rightarrow 1$ in 4-regular graphs when we increase $n$ as in our discussion in Subsection \ref{ss:regular}. Besides, $W({K}_n^c) \rightarrow 0.5$ as $n$ increases as discussed in Subsection \ref{ss:knc}. The commonly observed clustering of values near 0.5 may also present problems. Moreover, the model behind $W(G)$ is strange as signs of closed walks do not represent balance or imbalance. For these reasons we do not recommend $W(G)$ for future use.

The major weakness of the normalized algebraic conflict, $A(G)$, seems to be its incapability of evaluating the overall balance in graphs that have more than one connected component. Note that, some of the failures observed for $A(G)$ on axioms and desirable properties stem from its dependence on $\lambda(G)$ the smallest eigenvalue of the signed Laplacian matrix. $\lambda(G)$ might be determined by a component of the graph disconnected from other components and in turn not capturing the overall balance of the graph as a whole. For analysing graphs with just one cyclic connected component, one may use $A(G)$ while disregarding the acyclic components. However, if a graph has more than one cyclic connected component, using $A(G)$ or $\lambda(G)$ is similar to disregarding all but the most balanced connected component in the graph.

The three trivial measures, namely $X(G)$, $Y(G)$ and $Z(G)$, fail on various basic axioms and desirable properties in Tables \ref{tab3} and \ref{tab3.5}, and also show a lack of sensitivity to the graph, making them inappropriate to be used as measures of balance.

Satisfying almost all the axioms and desirable properties, $F(G)$ seems to measure something different from what is obtained using all cycles or all $k$-cycles, and be worth pursuing in future. Note that, $L(G)$ equals the minimum number of unbalanced fundamental cycles \cite{iacono_determining_2010}; suggesting a connection between the frustration and unbalanced cycles yet to be explored further. We recommend using $F(G)$ for all graphs as long as their size allows computing $L(G)$. Recently developed optimization models \cite{aref2017computing, aref2016exact} are shown to be capable of computing the frustration index in graphs of up to thousands of nodes and edges. For larger graphs, exact computation of $L(G)$ would be time consuming and it can be approximated using a nonzero optimality gap tolerance with the optimization models in \cite{aref2017computing, aref2016exact}. Alternatively, $A(G)$ and $D_k(G)$ seem to be the other options. Depending on the type of the graph, $k$-cycles might not necessarily capture global structural properties. For instance, this would make $D_3(G)$ an improper choice for some specific graphs like sparse 4-regular graphs (as in Subsection \ref{ss:regular}), square grids, and sparse graphs with a small number of 3-cycles. Similarly, $A(G)$ is not suitable for graphs that have more than one connected component (including many sparse graphs). 

\subsection*{\textbf{Notes on previous work}}

In the literature, balance theory is widely used on directed signed graphs. It seems that this approach is questionable in two ways. First, it neglects the fact that many edges in signed digraphs are not reciprocated. Bearing that in mind, investigating balance theory in signed digraphs deals with conflict avoidance when one actor in such a relationship may not necessarily be aware of good will or ill will on the part of other actors. This would make studying balance in directed networks analogous to studying how people avoid potential conflict resulting from potentially unknown ties. Secondly, balance theory does not make use of the directionality of ties and the concepts of sending and receiving positive and negative links. 

Leskovec, Huttenlocher and Kleinberg compare the reliability of predictions made by competing theories of social structure: balance theory and status theory (a theory that explicitly includes direction and gives quite different predictions) \cite{leskovec_signed_2010}. The consistency of these theories with observations is investigated through large signed directed networks such as Epinions, Slashdot, and Wikipedia. The results suggest that status theory predicts accurately most of the time while predictions made by balance theory are incorrect half of the time. This supports the inefficacy of balance theory for structural analysis of signed digraphs. For another comparison of the theories on signed networks, one may refer to a study of 8 theories to explain signed tie formation between students \cite{yap_why_2015}.

In a parallel line of research on network structural analysis, researchers differentiate between classical balance theory and structural balance specifically in the way that the latter is directional \cite{bonacich_introduction_2012}. They consider another setting for defining balance where absence of ties implies negative relationships. This assumption makes the theory limited to complete signed digraphs. Accordingly, 64 possible structural configurations emerge for three nodes. These configurations can be reduced to 16 classes of triads, referred to as 16 MAN triad census, based on the number of Mutual, Asymmetric, and Null relationships they contain. There are only 2 out of 16 classes that are considered balanced. New definitions are suggested by researchers in order to make balance theory work in a directional context. According to Prell \cite{prell_social_2012}, there is a second, a third, and a fourth definition of permissible triads allowing for 3, 7, and 9 classes of all 16 MAN triads. However, there have been many instances of findings in conflict with expectations \cite{prell_social_2012}.

Apart from directionality, the interpretation of balance measures is very important. Numerous studies have compared balance measures with their extremal values and found that signed networks are far from balanced, for example \cite{estrada_walk-based_2014}. However, with such a strict criterion, we must be careful not to look for properties that are almost impossible to satisfy. A much more systematic approach is to compare values of partial balance in the signed graphs in question to the corresponding values for reshuffled graphs \cite{Szell_multi, Szell_measure} as we have done in Section~\ref{s:real}. 

So far we formalized the notion of partial balance and compared various measures of balance based on their values in different graphs where the underlying structure was not important. We also evaluated the measures based on their axiomatic properties and ruled out the measures that we could not justify. In the next section, we focus on exploring real signed graphs based on the justified methods.

\section{Results on real signed networks} \label{s:real}

\begin{figure}[]
	\subfloat[Highland tribes network (G1), a signed network of 16 tribes of the Eastern Central Highlands of New Guinea \cite{read_cultures_1954}]{\includegraphics[width=2.5in]{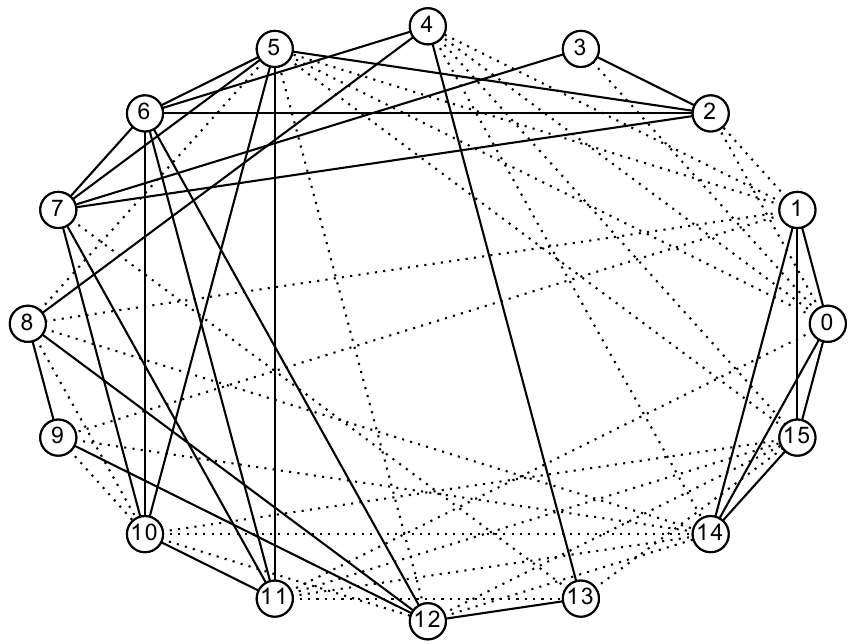}%
		\label{fig6_first_case}}
	\hfil
	\subfloat[Monastery interactions network (G2) of 18 New England novitiates inferred from the integration of all positive and negative relationships \cite{sampson_novitiate_1968}]{\includegraphics[width=2.5in]{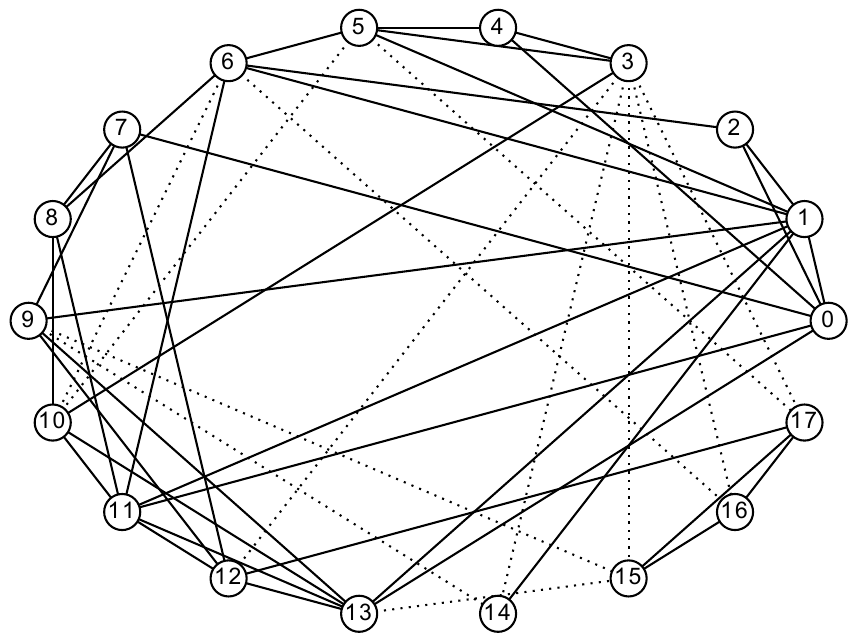}%
		\label{fig6_second_case}}
	~
	
	\subfloat[Fraternity preferences network (G3) of 17 boys living in a pseudo-dormitory inferred from ranking data of the last week in \cite{newcomb_acquaintance_1961}]{\includegraphics[width=2.5in]{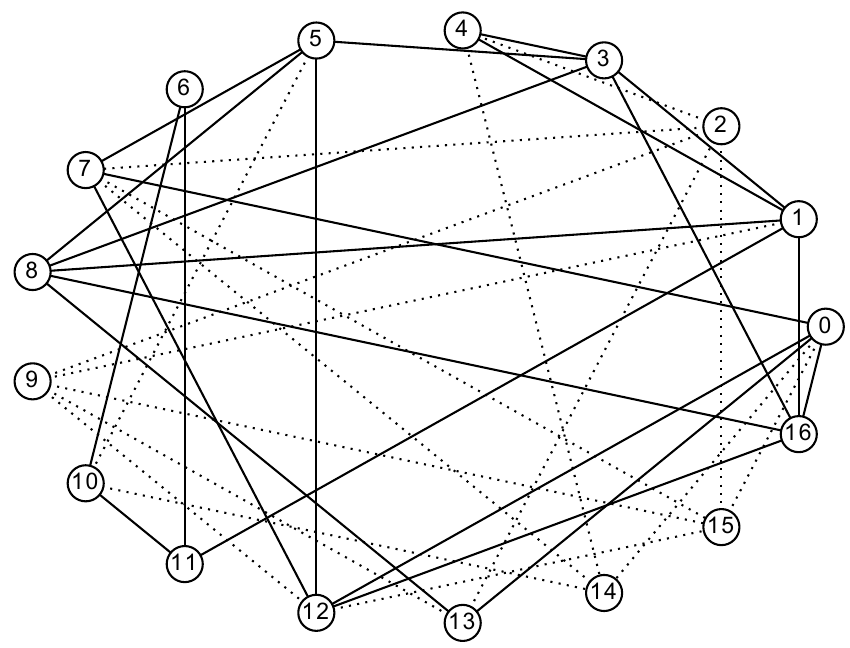}%
		\label{fig6_third_case}}
	\hfil
	\subfloat[College preferences network (G4) of 17 girls at an Eastern college inferred from ranking data of house B in \cite{lemann_group_1952}]{\includegraphics[width=2.5in]{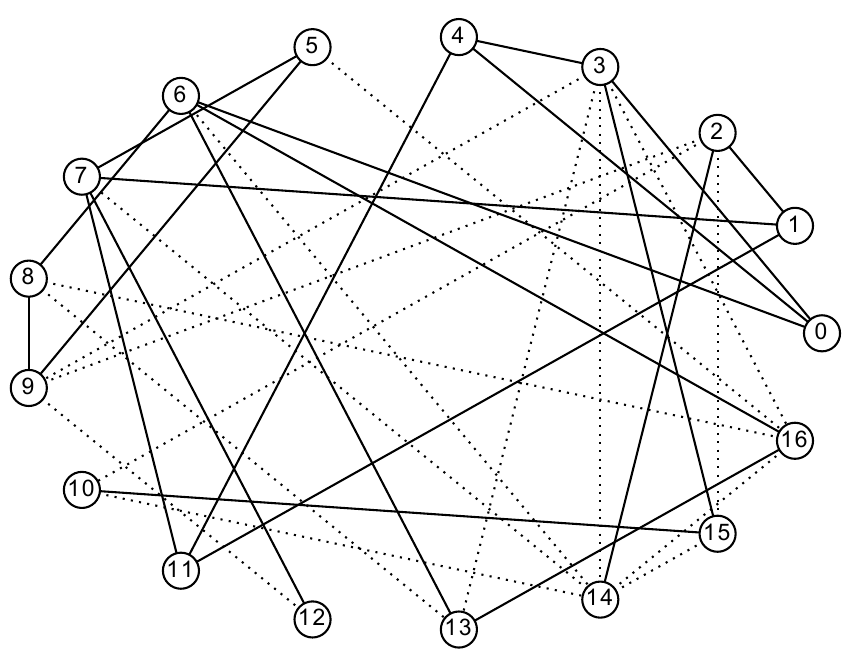}%
		\label{fig6_fourth_case}}
	\caption{Four small signed networks visualized where dotted lines represent negative edges and solid lines represent positive edges}
	\label{fig6}
\end{figure}

In this section, we analyze partial balance for a range of signed networks inferred from datasets of positive and negative interactions and preferences. Read's dataset for New Guinean highland tribes \cite{read_cultures_1954} is demonstrated as a signed graph (G1) in Figure~\ref{fig6}(a), where dotted lines represent negative edges and solid lines represent positive edges. The fourth time window of Sampson's dataset for monastery interactions \cite{sampson_novitiate_1968} (G2) is drawn in Figure~\ref{fig6}(b). We also consider datasets of students' choice and rejection (G3 and G4) \cite{newcomb_acquaintance_1961,lemann_group_1952} as demonstrated in Figure~\ref{fig6}(c) and Figure~\ref{fig6}(d). The last three are converted to undirected signed graphs by considering mutually agreed relations. A further explanation on the details of inferring signed graphs from the choice and rejection data is provided in Appendix \ref{s:infer}. 

A larger signed network (G5) is inferred by \cite{neal_backbone_2014} through implementing a stochastic degree sequence model on Fowler's data on Senate bill co-sponsorship \cite{fowler_legislative_2006}. Besides the signed social network datasets, large scale biological networks can be analysed as signed graphs. There are relatively large signed biological networks analysed by \cite{dasgupta_algorithmic_2007} and \cite{iacono_determining_2010} from a balance viewpoint under a different terminology where \textit{monotonocity} is the equivalent for balance. The two gene regulatory networks we consider are related to two organisms: a eukaryote (the yeast \textit{Saccharomyces cerevisiae}) and a bacterium (\textit{Escherichia coli}). Graphs G6 and G7 represent the gene regulatory networks of \textit{Saccharomyces cerevisiae} \cite{Costanzo2001yeast} and \textit{Escherichia coli} \cite{salgado2006ecoli} respectively. Note that, the densities of these networks are much smaller than the other networks introduced above. In gene regulatory networks, nodes represent genes. Positive and negative edges represent \textit{activating connections} and \textit{inhibiting connections} respectively. Figure~\ref{fig6.5} shows the bill co-sponsorship network as well as biological signed networks. The color of edges correspond to the signs on the edges (green for $+1$ and red for $-1$). For more details on the biological datasets, one may refer to \cite{iacono_determining_2010}.

\begin{figure}[]
	\centering
	\subfloat[The bill co-sponsorship network (G5) of senators \cite{neal_backbone_2014}]{\includegraphics[width=3.5in]{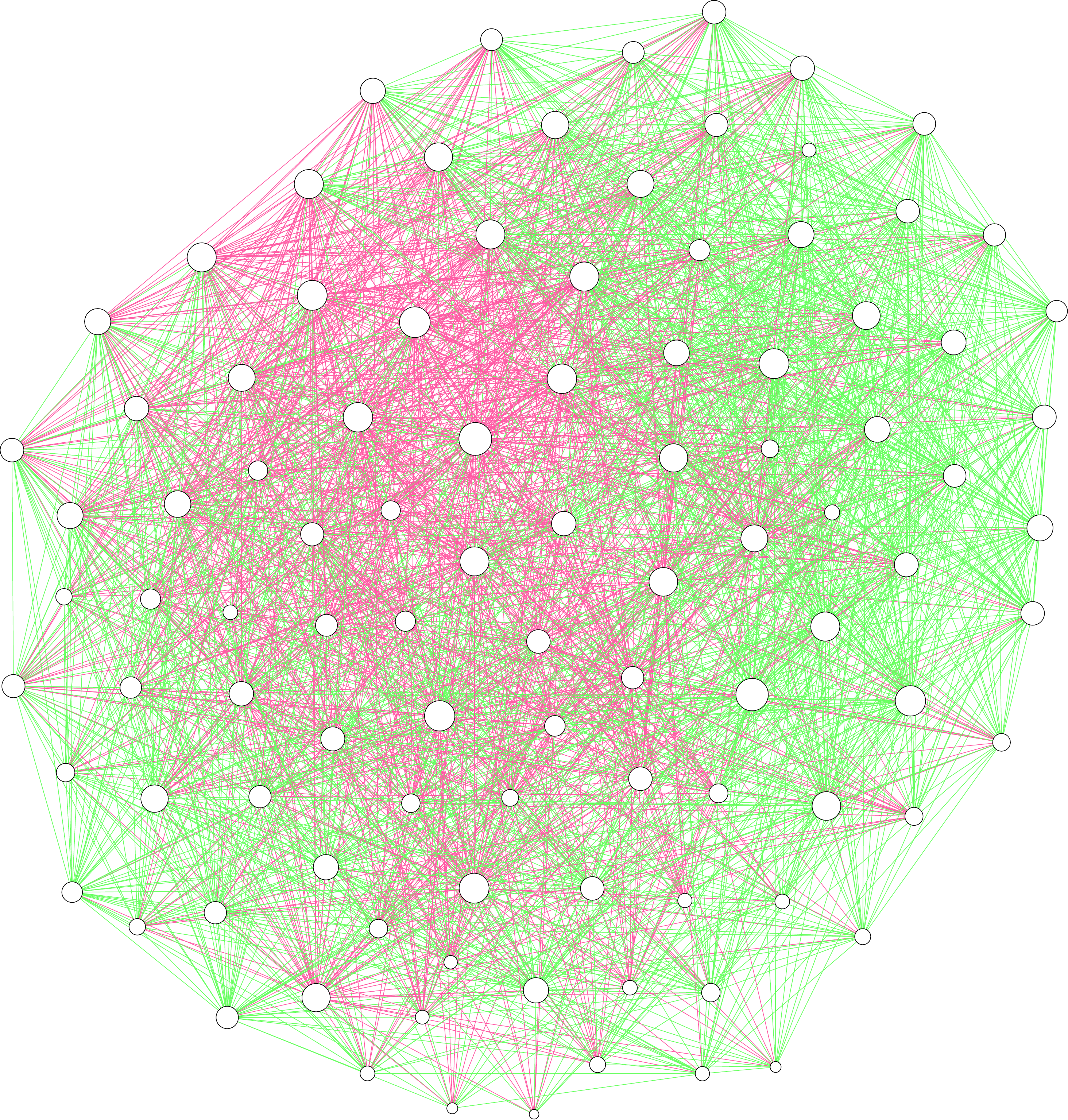}%
		\label{fig_first_case}}
	
	~
	
	\subfloat[The gene regulatory network (G6) of \textit{Saccharomyces cerevisiae} \cite{Costanzo2001yeast}]{\includegraphics[width=2.5in]{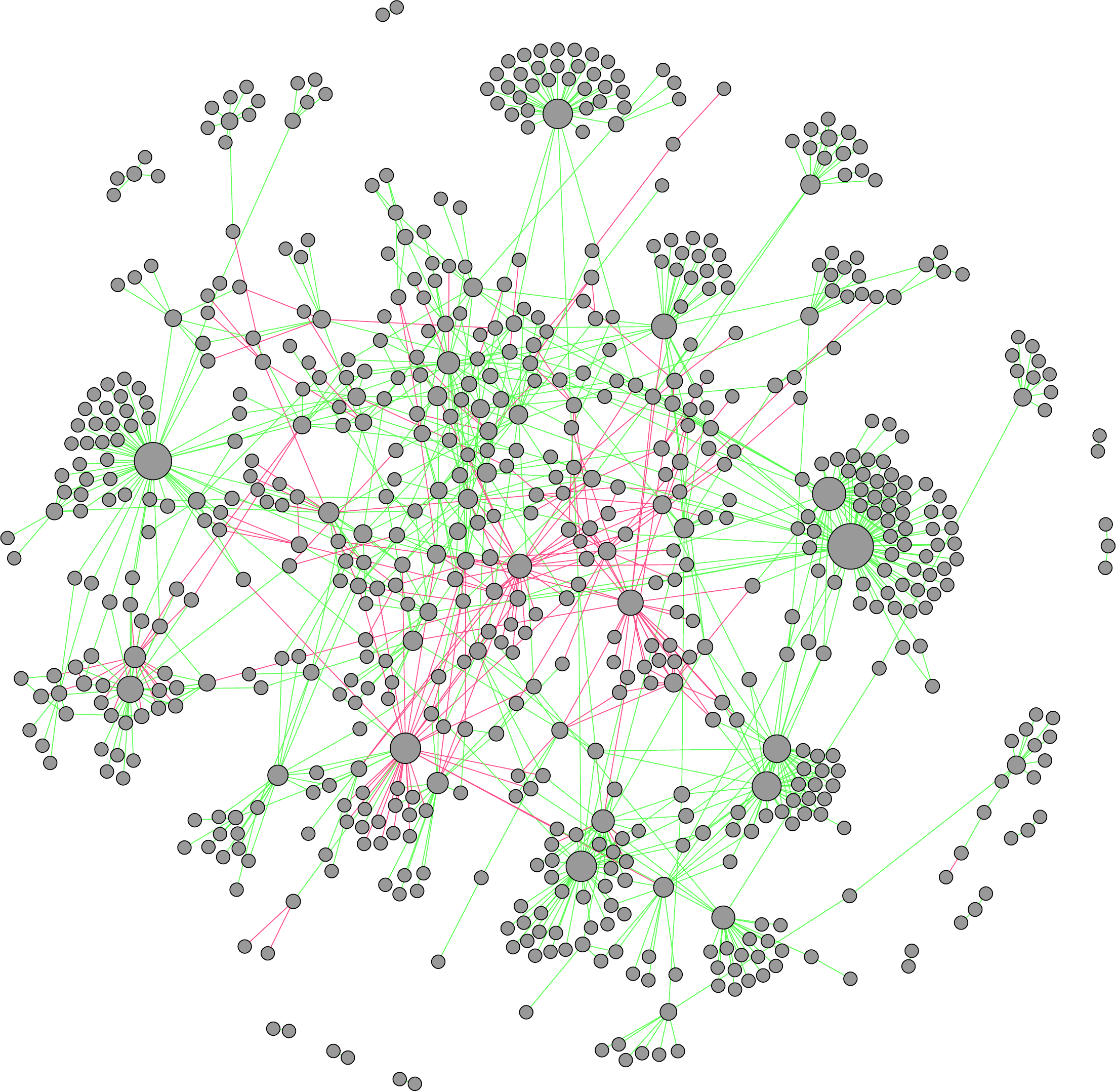}%
		\label{fig_second_case}}
	\hfil
	\subfloat[The gene regulatory network (G7) of the \textit{Escherichia coli} \cite{salgado2006ecoli}]{\includegraphics[width=2.5in]{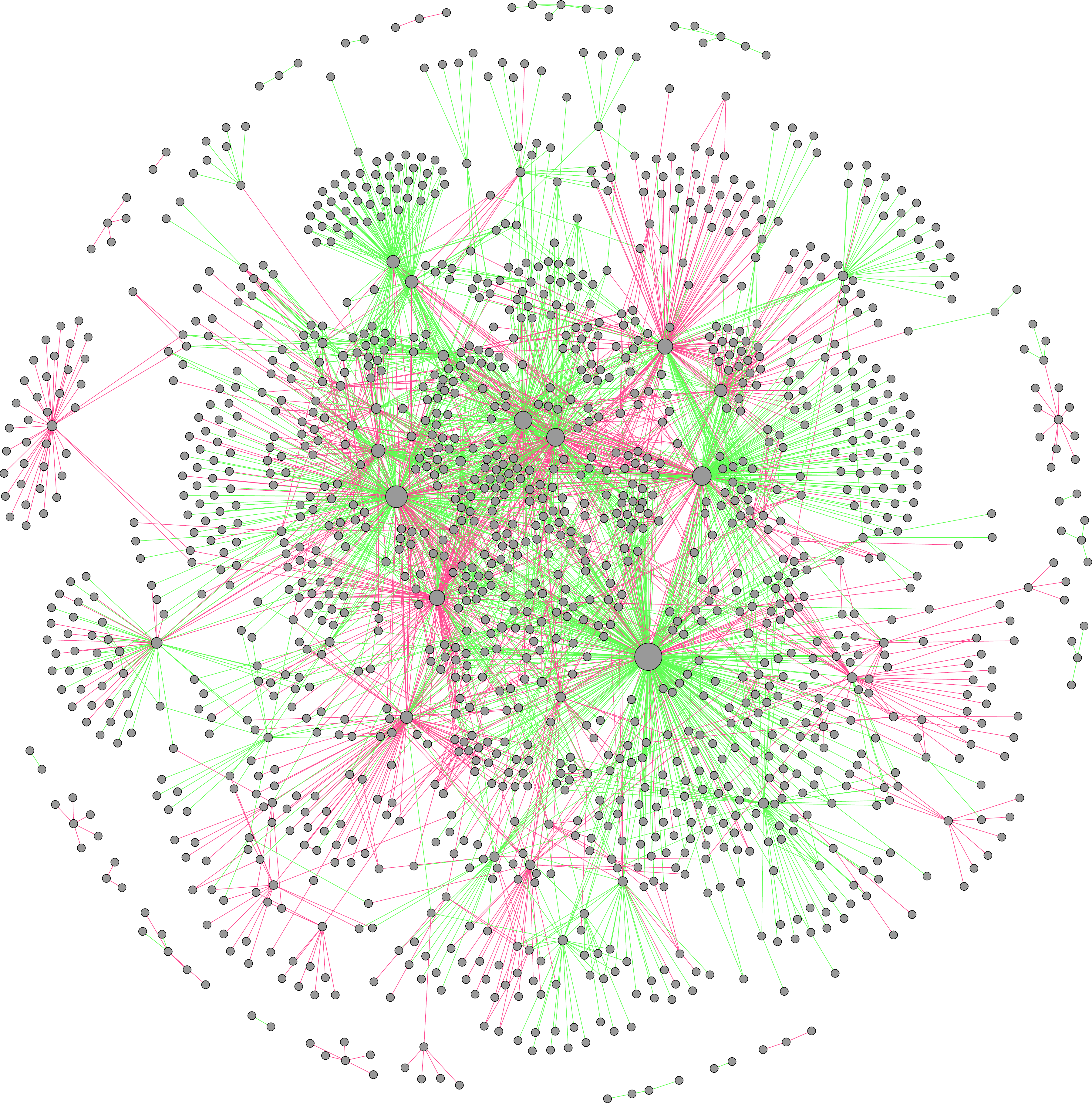}%
		\label{fig_third_case}}
	
	\caption{Three larger signed datasets illustrated as signed graphs in which red lines represent negative edges and green lines represent positive edges (color version online)}

	\label{fig6.5}
\end{figure}
As Figure~\ref{fig6.5} shows, graphs G6 and G7 have more than one connected component. Besides the giant component, there are a number of small components that we discard in order to use $A(G)$ and $\lambda(G)$. Note that, this procedure does not change $T(G)$ and $L(G)$ as the small components are all acyclic. The values of $(n, m, m^-)$ for giant components of G6 and G7 are $(664, 1064, 220)$ and $(1376, 3150, 1302)$ respectively.

The results are shown in Table~\ref{tab4}. Although neither of the networks is completely balanced, the small values of $L(G)$ suggest that removal of relatively few edges makes the networks completely balanced. Table~\ref{tab4} also provides a comparison of partial balance between different datasets of similar sizes. In this regard, it is essential to know that the choice of measure can make a substantial difference. For instance among G1--G4, under $T(G)$, G1 and G3 are respectively the most and the least partially balanced networks. However, if we choose $A(G)$ as the measure, G1 and G3 would be the least and the most partially balanced networks respectively. This confirms our previous discussions on how choosing a different measure can substantially change the results and helps to clarify some of the conflicting observations in the literature \cite{facchetti_computing_2011,kunegis_applications_2014} and \cite{estrada_walk-based_2014}, as previously discussed in Section~\ref{s:recom}.

\begin{table}[ht]
	\centering
	\caption{Partial balance computer for signed graphs (G1--7) and reshuffled graphs}
	\label{tab4}
	\begin{tabularx}{1\textwidth}{p{2.5cm}p{0.2cm}p{1.2cm}p{0.4cm}p{0.4cm}p{0.4cm}p{0.4cm}p{0.4cm}}
		\hline
		Graph:$(n, m, m^-)$                    & Density                &                         & $T$   & $A$   & $F$   & $\lambda$ & $L$    \\ \hline
		\multirow{4}{*}{G1: (16, 58, 29)}       & \multirow{4}{*}{0.483} & $\mu(G)$                & 0.87  & 0.88  & 0.76  & 1.04      & 7      \\
		&                        & $\text{mean}(\mu(G_r))$ & 0.50  & 0.76  & 0.49  & 2.08      & 14.65  \\
		&                        & $\text{SD}(\mu(G_r))$   & 0.06  & 0.02  & 0.05  & 0.20      & 1.38   \\
		&                        & Z-score                 & 6.04  & 5.13  & 5.54  & $-5.13$     & $-5.54$  \\ \hline
		\multirow{4}{*}{G2: (18, 49, 12)}       & \multirow{4}{*}{0.320} & $\mu(G)$                & 0.86  & 0.88  & 0.80  & 0.75      & 5      \\
		&                        & $\text{mean}(\mu(G_r))$ & 0.55  & 0.79  & 0.60  & 1.36      & 9.71   \\
		&                        & $\text{SD}(\mu(G_r))$   & 0.09  & 0.03  & 0.05  & 0.18      & 1.17   \\
		&                        & Z-score                 & 3.34  & 3.37  & 4.03  & $-3.37$     & $-4.03$  \\ \hline
		\multirow{4}{*}{G3:(17, 40, 17)}       & \multirow{4}{*}{0.294} & $\mu(G)$                & 0.78  & 0.90  & 0.80  & 0.50      & 4      \\
		&                        & $\text{mean}(\mu(G_r))$ & 0.49  & 0.82  & 0.62  & 0.89      & 7.53   \\
		&                        & $\text{SD}(\mu(G_r))$   & 0.11  & 0.06  & 0.06  & 0.30      & 1.24   \\
		&                        & Z-score                 & 2.64  & 1.32  & 2.85  & $-1.32$     & $-2.85$  \\ \hline
		\multirow{4}{*}{G4: (17, 36, 16)}       & \multirow{4}{*}{0.265} & $\mu(G)$                & 0.79  & 0.88  & 0.67  & 0.71      & 6      \\
		&                        & $\text{mean}(\mu(G_r))$ & 0.49  & 0.87  & 0.64  & 0.79      & 6.48   \\
		&                        & $\text{SD}(\mu(G_r))$   & 0.14  & 0.03  & 0.06  & 0.17      & 1.08   \\
		&                        & Z-score                 & 2.16  & 0.50  & 0.45  & $-0.50$     & $-0.45$  \\ \hline
		\multirow{4}{*}{G5: (100, 2461, 1047)}  & \multirow{4}{*}{0.497} & $\mu(G)$                & 0.86  & 0.87  & 0.73  & 8.92      & 331    \\
		&                        & $\text{mean}(\mu(G_r))$ & 0.50  & 0.75  & 0.22  & 17.46     & 965.6  \\
		&                        & $\text{SD}(\mu(G_r))$   & 0.00  & 0.00  & 0.01  & 0.02      & 9.08   \\
		&                        & Z-score                 & 118.5 & 387.8 & 69.89 & $-387.8$    & $-69.89$ \\ \hline
		\multirow{4}{*}{G6:(690, 1080, 220)}   & \multirow{4}{*}{0.005} & $\mu(G)$                & 0.54  & 1.00  & 0.92  & 0.02      & 41     \\
		&                        & $\text{mean}(\mu(G_r))$ & 0.58  & 1.00  & 0.77  & 0.02      & 124.3  \\
		&                        & $\text{SD}(\mu(G_r))$   & 0.07  & 0.00  & 0.01  & 0.00      & 4.97   \\
		&                        & Z-score                 & $-0.48$ & 8.61  & 16.75 & $-8.61$     & $-16.75$ \\ \hline
		\multirow{4}{*}{G7:(1461, 3215, 1336)} & \multirow{4}{*}{0.003} & $\mu(G)$                & 0.50  & 1.00  & 0.77  & 0.06      & 371    \\
		&                        & $\text{mean}(\mu(G_r))$ & 0.50  & 1.00  & 0.59  & 0.06      & 653.4  \\
		&                        & $\text{SD}(\mu(G_r))$   & 0.02  & 0.00  & 0.00  & 0.00      & 7.71   \\
		&                        & Z-score                 & $-0.33$ & 3.11  & 36.64 & $-3.11$     & $-36.64$ \\ \hline
				
	\end{tabularx}
\end{table}

In Table~\ref{tab4}, the mean and standard deviation of measures for the reshuffled graphs $(G_r)$, denoted by $\text{mean}(\mu(G_r))$ and $\text{SD}(\mu(G_r))$, are also provided for comparison. We implement a very basic statistical analysis as in \cite{Szell_multi, Szell_measure} using $\text{mean}(\mu(G_r))$ and $\text{SD}(\mu(G_r))$ of 500 reshuffled graphs. Reshuffling the signs on the edges 500 times, we obtain two parameters of balance distribution for the fixed underlying structure. For measures of balance, Z-scores are calculated based on Equation \eqref{eq26}.
\begin{align}
\label{eq26}
Z=\frac{\mu(G)-\text{mean}(\mu(G_r))}{\text{SD}(\mu(G_r))}
\end{align}

The Z-score shows how far the balance is with regards to balance distribution of the underlying structure. Positive values of Z-score for $T(G)$, $A(G)$, and $F(G)$ can be interpreted as existence of more partial balance than the average random level of balance. 

It is worth pointing out that the statistical analysis we have implemented is independent of the normalization method used in $A(G)$ and $F(G)$. The two right columns of \ref{tab4} provide $\lambda(G)$ and $L(G)$ alongside their associated Z-scores. 

The Z-scores show that as measured by the frustration index and algebraic conflict, signed networks G1--G7 exhibit a level of partial (but not total) balance beyond what is expected by chance. Based on these two measures, the level of partial balance is high for graphs G1, G2, G5, G6, and G7 while the numerical results for G3 and G4 do not allow a conclusive interpretation. It indicates that most of the real signed networks investigated are relatively consistent with the theory of structural balance. However, the Z-scores obtained based on the triangle index for G6--G7 show totally different results. Note that G6 and G7 are relatively sparse graphs which only have 70 and 1052 triangles. This may explain the difference between Z-scores of $T(G)$ and that of other measures. The numerical results using the algebraic conflict and frustration index support previous observations of real-world networks' closeness to balance \cite{facchetti_computing_2011,kunegis_applications_2014}.

\section{Conclusion and Future research} \label{s:conclu}

In this study, we started by discussing balance in signed networks in Sections~\ref{s:problem} and \ref{s:check} and introduced the notion of partial balance. We discussed different ways to measure partial balance in Section~\ref{s:measure} and provided some observations on synthetic data in Sections~\ref{s:basic} and \ref{s:special}. After gaining an understanding of the behaviour of different measures, basic axioms and desirable properties were used in Section~\ref{s:axiom} to rule out the measures that cannot be justified.

We have discussed various methodologies and how they have led to conflicting observations in the literature in Section~\ref{s:recom}. Taking axiomatic properties of the measures into account, using the common cycle-based measures denoted by $D(G)$ and $C(G)$ and the walk-based measure $W(G)$ is not recommended. $D_k(G)$ and $A(G)$ may introduce some problems, but overall using them seems to be more appropriate compared to $D(G)$, $C(G)$ and $W(G)$. The observations on synthetic data taken together with the axiomatic properties, recommend $F(G)$ as the best overall measure of partial balance. However, considering the difficulty of computing the exact value of $L(G)$ for very large graphs, one may approximate it using a nonzero optimality gap tolerance with the optimization models in \cite{aref2017computing, aref2016exact}. Alternatively, $A(G)$ and $D_k(G)$ seem to be the other options accepting their potential shortcomings.

Using the three measures $F(G)$, $T(G)$, and $A(G)$, each representing a family of measures, we compared balance in real signed graphs and analogous reshuffled graphs having the same structure in Section~\ref{s:real}. Table~\ref{tab4} provides this comparison showing that different results are obtained under different measures. 

Returning to the questions posed at the beginning of this paper, it is now possible to state that under the frustration index and algebraic conflict many signed networks exhibit a level of partial (but not total) balance beyond that expected by chance. However, the numerical results in Table~\ref{tab4} show that the level of balance observed using the triangle index can be totally different. One of the more significant findings to emerge from this study is that methods suggested for measuring balance may have different context and may require some justification before being interpreted based on their values. The present study confirms that some measures of partial balance cannot be taken as a reliable static measure to be used for analyzing network dynamics. 

Although a numerical part of the current study is based on signed networks with less than a few thousand nodes, the analytical findings that were not restricted to a particular size suggest the inefficacy of some methods for analyzing larger networks as well.

One gap in this study is that we avoid using structural balance theory for analyzing directed networks, making directed signed networks like Epinions, Slashdot, and Wikipedia Elections \cite{leskovec_signed_2010,estrada_walk-based_2014,Giscard2016} datasets untested by our approach. However, see our discussion in Section~\ref{s:recom}.

The findings of this study have a number of important implications for future investigation. Although this study focuses on partial balance, the findings may well have a bearing on link prediction and clustering in signed networks \cite{Gallier16}. Some other theoretical topics of interest in signed networks are network dynamics \cite{tan_evolutionary_2016} and opinion dynamics \cite{li_voter_2015}. Effective methods of signed network structural analysis can contribute to these topics as well. From a practical viewpoint, international relations is a crucial area to implement signed network structural analysis. Having an efficient measure of partial balance in hand, international relations can be investigated in terms of evaluation of partial balance over time for networks of states. 

\section*{Acknowledgement}
We are grateful for the valuable comments from the anonymous reviewers that have improved this paper.

\appendix
\section{Details of calculations} \label{s:calc}
In order to simplify the sum $E(D_k(G)) = \sum \limits_{i \: \text{even}} {\binom{k}{i}} q^i (1-q)^{k-i}$, one may add the two following equations and divide the result by 2:
\begin{align}
	\sum \limits_{i} {\binom{k}{i}} q^i (1-q)^{k-i} &= (q+(1-q))^k \\
	\sum \limits_{i} {\binom{k}{i}} (-q)^i (1-q)^{k-i} &= (-q+(1-q))^k
\end{align}

In ${K}_n^a$, a $k$-cycle is specified by choosing $k$ vertices in some order, then correcting for the overcounting by dividing by $2$ (the possible directions) and $k$ (the number of starting points, namely the length of the cycle). If the unique negative edge is required to belong to the cycle, by orienting this in a fixed way we need choose only $k-2$ further elements in order, and no overcounting occurs. The numbers of negative cycles and total cycles are as follows.

\begin{align}
\sum_{k=3}^n O_k^- 
=\sum_{k=3}^n\frac{(n-2)!}{(n-k)!}
,\quad
\sum_{k=3}^n O_k 
=\sum_{k=3}^n \frac{n!}{2k(n-k)!}
\end{align}

Asymptotic approximations for these sums can be obtained by introducing the exponential generating function. For example, letting 
$$
a_n = \sum_{1\leq k\leq n} \frac{n!}{(n-k)!k}
$$
we have
$$
\sum_{n\geq 0} \frac{a_n}{n!}x^n = \sum_{n,k} \sum_{k\leq n} \frac{n!}{(n-k)!k} x^n = \sum_{k\geq 1} \frac{1}{k} \sum_{n\geq k} \frac{1}{(n-k)!} x^n = 
\sum_{k\geq 1} \frac{x^k}{k} \sum_{m\geq 0} \frac{1}{m!}x^m = e^x \log\left(\frac{1}{1-x}\right).
$$
Similarly we obtain
$$
\sum_{n\geq 0} \sum_{k\geq 0} \frac{1}{(n-k)!} x^n = \frac{e^x}{1-x}.
$$
Standard singularity analysis methods \cite{flajolet2009analytic} show the denominator of the expression for $D(K^a_n)$ to be asymptotic to $(n-1)!e/2$ while the number of negative cycles is asymptotic to $(n-2)!e$. Similarly the weighted sum defining $C(K^a_n)$, where we choose $f(k) = 1/k!$, can be expressed using the ordinary generating function, which for the denominator turns out to be 
$$\frac{1}{1-x} \log\left( \frac{1-2x}{1-x}\right).$$ Again, singularity analysis techniques yield an approximation $2^n/n$. The numerator is easier, and asymptotic to $2^n/(n^2-n)$. This yields the result.

The unsigned adjacency matrix $|\textbf{A}|$ of the complete graph has the form $\textbf{E} - \textbf{I}$ where $\textbf{E}$ is the matrix of all $1$'s. The latter matrix has rank 1 and nonzero eigenvalue $n$. Thus $|\textbf{A}| _{({K}_n^a)}$ has eigenvalues $n-1$ (with multiplicity 1) and $-1$ (with multiplicity $n-1$). The matrix $\textbf{A} _{({K}_n^a)}$ has a similar form and we can guess eigenvectors of the form $(-1,1,0, \dots, 0)$ and $(a,a,1,1, \dots , 1)$. Then $a$ satisfies a quadratic $2a^2 + (n-3)a - (n-2) = 0$. Solving for $a$ and the corresponding eigenvalues, we obtain eigenvalues $(n-4 \pm \sqrt{(n-2)(n+6)})/2, 1, -1$ (with multiplicity $n-3$)).

This yields 
$$K{({K}_n^a)} = \frac{(n-3)e^{-1} + e + e^{\frac{n-4 - \sqrt{(n-2)(n+6)}}{2}} + e^{\frac{n-4 + \sqrt{(n-2)(n+6)}}{2}}}{(n-1)e^{-1} + e^{n-1}}$$ which results in $W({K}_n^a) \sim \frac{1+e^{-4/n}}{2}$. 

Furthermore, since every node of $K_n$ has degree $n-1$, the eigenvalues of $\textbf{L}:=(n-1)\textbf{I}-\textbf{A}$ are precisely of the form $n-1 - \lambda$ where $\lambda$ is an eigenvalue of $\textbf{A}$.

Measures of partial balance for ${K}_n^a$ can therefore be expressed by the formulae \eqref{eq15} -- \eqref{eq18}:

\begin{align}\label{eq15}
D({K}_n^a)=1- \frac{\sum_{k=3}^n \frac{(n-2)!}{(n-k)!}}{\sum_{k=3}^n \frac{n!}{2k(n-k)!}} \sim 1 - \frac{2}{n}
\end{align}
\begin{align}\label{eq15.5}
C({K}_n^a)=1- \frac{\sum_{k=3}^n \frac{(n-2)!}{(n-k)!k!}}{\sum_{k=3}^n \frac{n!}{2k(n-k)!k!}} \sim 1 - \frac{1}{n}
\end{align}
\begin{align}\label{eq17}
D_k({K}_n^a)=1-\frac{\frac{(n-2)!}{(n-k)!}}{\frac{n!}{2k(n-k)!}}=1-\frac{2k}{n(n-1)} \sim 1 - \frac{2k}{n^2}
\end{align}
\begin{align}\label{eq16}
W({K}_n^a) \sim \frac{1+e^{-4/n}}{2} \sim 1-\frac{2}{n}
\end{align}
\begin{align}\label{eq17.2}
\lambda({K}_n^a)= n-1 - (n-4 + \sqrt{(n-2)(n+6)})/2 = (n+2 - \sqrt{(n-2)(n+6)})/2
\end{align}
\begin{align}\label{eq17.4}
A({K}_n^a)= 1- \frac{n+2 - \sqrt{(n-2)(n+6)}}{2n-4} \sim 1-\frac{4}{n^2}
\end{align}
\begin{align}\label{eq18}
F({K}_n^a)=1-\frac{2}{n(n-1)/2}=1-\frac{4}{n(n-1)} \sim 1 - \frac{4}{n^2}
\end{align}

In ${K}_n^c$, all cycles of odd length are unbalanced and all cycles of even length are balanced. Therefore:
\begin{align}
\sum_{k=3}^n O_k^+ 
= \sum_{\textnormal{even}}^n \frac{n!}{2k(n-k)!}
\end{align}

It follows that $D_k({K}_n^c)$ equals 0 for odd $k$ and 1 for even $k$. Based on maximality of $\lambda(G)$ in ${K}_n^c$, $A({K}_n^c)=0$.

Using the above generating function techniques we obtain that the numerator of $D(K^c_n)$ is asymptotic to $(n-1)!(e+(-1)^ne^{-1})/4$. The denominator we know from above is asymptotic to $(n-1)!e/2$. This yields $D(K^c_n) \sim 1/2 + (-1)^n e^{-2}$. Note that $e^{-2} \approx 0.135$ and this explains the oscillation in Figure~\ref{fig5}. Similarly we obtain results for $C(K^c_n)$.

$|\textbf{A}|_{({K}_n^c)}$ has eigenvalues $n-1$ (with multiplicity 1) and $-1$ (with multiplicity $n-1$). The matrix $\textbf{A}_{({K}_n^c)}$ has a similar form and the corresponding eigenvalues would be $1-n$ (with multiplicity 1) and $1$ (with multiplicity $n-1$). This yields
$K{({K}_n^c)} = \frac{(n-1)e^{1} + e^{1-n}}{(n-1)e^{-1} + e^{n-1}}$ which results in $W({K}_n^c) \sim \frac{1+e^{2-2n}}{2}$.

Moreover, a closed-form formula for $L({K}_n^c)$ can be expressed based on a maximum cut which gives a function of $n$ equal to an upper bound of the frustration index under a different name in \cite{abelson_symbolic_1958}. Measures of partial balance for ${K}_n^c$ can be expressed via the closed-form formulae as stated in \eqref{eq21} -- \eqref{eq25}:

\begin{align}\label{eq21}
D({K}_n^c)= \frac{\sum_{\textnormal{$k$ even}}^n \frac{n!}{2k(n-k)!}}{\sum_{k=3}^n \frac{n!}{2k(n-k)!}} \sim \frac{1}{2} + (-1)^n e^{-2}
\end{align}
\begin{align}\label{eq22}
C({K}_n^c)= \frac{\sum_{\textnormal{$k$ even}}^n \frac{n!}{2k(n-k)!k!}}{\sum_{k=3}^n \frac{n!}{2k(n-k)!k!}} \sim \frac{1}{2} - \frac{3n\log n}{2^n}
\end{align}
\begin{align} \label{eq22.5}
D_k({K}_n^c)=
\left\{
\begin{array}{ll}
1 & \mbox{if } k \ \text{is even} \\
0 & \mbox{if } k \ \text{is odd}
\end{array}
\right.
\end{align}
\begin{align}\label{eq23}
W({K}_n^c) \sim \frac{1+e^{2-2n}}{2}
\end{align}
\begin{align}\label{eq23.3}
\lambda({K}_n^c)=\lambda_\text{max}=\overline{d}_{\text{max}}-1=n-2
\end{align}
\begin{align} \label{eq24}
L({K}_n^c)=
\left\{
\begin{array}{ll}
({n^2-2n})/{4} & \mbox{if } n  \ \text{is even} \\
({n^2-2n+1})/{4} & \mbox{if } n  \ \text{is odd}
\end{array}
\right.
\end{align}
\begin{align} \label{eq25}
F({K}_n^c)=
\left\{
\begin{array}{ll}
1-\frac{n(n-2)/4}{n(n-1)/4}=\frac{1}{n-1} & \mbox{if } n  \ \text{is even} \\
1-\frac{(n-1)(n-1)/4}{n(n-1)/4}=\frac{1}{n} & \mbox{if } n  \ \text{is odd}
\end{array}
\right.
\end{align}

Our calculations for $L({K}_n^c)$ show that the upper bound suggested for the frustration index in \cite{iacono_determining_2010} is incorrect.

\section{Counterexamples and proof ideas for the axioms and desirable properties} \label{s:counter}

Axioms: 

Axiom 1 holds in all the measures introduced due to the systematic normalization implemented.

$D_k(G)$, $A(G)$, and $Y(G)$ do not satisfy Axiom 2. All $k$-cycles being balanced, $D_k(G)$ fails to detect the imbalance in graphs with unbalanced cycles of different length. $A(G \oplus C^+)=1$ for unbalanced graphs which makes $A(G)$ fail Axiom 2. $Y(G)$ fails on detecting balance in completely bi-polar signed graphs that are indeed balanced.

As long as $\mu(G\oplus H)$ can be written in the form of $(a+c)/(b+d)$ where $\mu(G)=a/b$ and $\mu(H)=c/d$, $\mu$ satisfies Axiom 3. So all the measures considered satisfy Axiom 3, except for $A(G)$. In case of $\lambda(G) < \lambda(H)$ and $\lambda_\text{max}(G) < \lambda_\text{max}(H)$, $A(G \oplus H)= 1 - \frac{\lambda(G)}{\lambda_\text{max}(H)} > A(H)$ which shows that $A(G)$ fails Axiom 3.

The sign of cycles (closed walks), the Laplacian eigenvalues \cite{Belardo2016}, and the frustration index \cite{zaslavsky_matrices_2010} will not change by applying the switching function introduced in \eqref{eq1.2}. Therefore, Axiom 4 holds for all the measures discussed except for $X(G)$ and $Y(G)$ because they depend on $m^-$, which changes in switching.\\

\noindent
Desirable properties:

Clearly in B1, $C^+_3$ contributes positively to $D(G)$ and $C(G)$, whereas for $D_k(G)$ it depends on $k$ which makes it fail B1. As $W(C^+_3)$ equals 1, $\Tr(e^{\textbf{A}})/\Tr(e^{|\textbf{A}|})$ would be added by equal terms in both the numerator and denominator leading to $W(G)$ satisfying B1. $A(G)$ satisfies B1 because $A(G\oplus C^+_3)=1$. As $m$ increases by 3, $F(G)$ satisfies B1. The dependency of $X(G)$ and $Y(G)$ on $m^-$ and incapability of the binary measure, $Z(G)$, in providing values between 0 and 1 make them fail B1.

$C^-_3$ adds only to the denominators of $D(G)$ and $C(G)$, whereas for $D_k(G)$ it depends on $k$ which makes it fail B2. Following the addition of a negative 3-cycle, $W(G)$ is observed to increase resulting in its failure in B2 (for example, take $G=K_5$ with single negative edge, and $C^-_3$ having a single negative edge). As $A(G)\neq0$ and $A(C^-_3)=0$, the value of $A(G \oplus C^-_3)$ does not change when a negative 3-cycle is added. Therefore, it fails B2.	
Moreover, $F(G)$ fails B2 whenever $L(G) \geq m/3$ as observed in a family of graphs in Subsection \ref{ss:knc}. However, $F^\prime (G)$ introduced in Eq.\ \ref{eq5.8} which only differs in normalization, satisfies this desirable property. The measures $X(G)$ and $Y(G)$ fail B2, but the binary measure, $Z(G)$, satisfies it.
	
All the cycle-based measures, namely $D(G),C(G)$, and $D_3(G)$ fail B3 (for example, take $G=K_4$ with two symmetrically located negative edges). $W(G)$ is also observed to fail B3 (for instance, take $G$ as the disjoint union of a 3-cycle and a 5-cycle each having 1 negative edge). It is known that $\lambda(G \ominus e) \leq \lambda(G)$ \cite{Belardo2016}. However in some cases where $\lambda_\text{max}(G \ominus e) < \lambda_\text{max} (G)$ counterexamples are found showing $A(G)$ fails on B3 (consider a graph with $n=8, m^+=10, m^-=3, \min{|E^*|}=3$ in which $\lambda_\text{max} (G)=6$ and $\lambda_\text{max}(G \ominus e)=3$). $Y(G)$ and $Z(G)$ fail B3. Moreover, $F(G)$ satisfies B3 because $L(G \ominus e) = L(G) -1$.

The cycle-based measures and $W(G)$ do not satisfy B4. For $D_3(G)$, we tested a graph with $n=7, m=15, |E^*|=3$ and we observed $D_3(G \oplus e) > D_3(G)$. According to Belardo and Zhou, $\lambda(G \oplus e) \geq \lambda(G)$ \cite{Belardo2016}. However in some cases where $\lambda_\text{max}(G \oplus e) > \lambda_\text{max} (G)$ counterexamples are found showing $A(G)$ fails on B4. counterexamples showing $D(G),C(G),W(G)$, and $A(G)$ fail B4, are similar to that of B3. Moreover, $F(G)$ satisfies B4 as do $X(G)$ and $Z(G)$, while $Y(G)$ fails B4 when $e$ is positive.

\section{Inferring undirected signed graphs} \label{s:infer}
Sampson collected different sociometric rankings from a group of 18 monks at different times \cite{sampson_novitiate_1968}. The data provided includes rankings on like, dislike, esteem, disesteem, positive influence, negative influence, praise, and blame. We have considered all positive and negative rankings. Then only the reciprocated relations with similar signs are considered to infer an undirected signed edge between two monks (see \cite{doreian_partitioning_2009} and how the authors inferred a directed signed graph in their Table 5 by summing the influence, esteem and respect relations).

Newcomb reported rankings made by 17 men living in a pseudo-dormitory \cite{newcomb_acquaintance_1961}. We used the ranking data of the last week which includes complete ranks from 1 to 17 gathered from each man. As the gathered data is related to complete ranking, we considered ranks 1-5 as one-directional positive relations and 12-17 as one-directional negative relations. Then only the reciprocated relations with similar signs are considered to infer an undirected signed edge between two men (see \cite{doreian_partitioning_2009} and how the authors converted the top three and bottom three ranks to a directed signed edges in their Fig. 4.).

Lemann and Solomon collected ranking data based on multiple criteria from female students living in off-campus dormitories \cite{lemann_group_1952}. We used the data for house B which is resulted by integrating top and bottom three rankings for multiple criteria. As the gathered data itself is related to top and bottom rankings, we considered all the ranks as one-directional signed relations. Then only the reciprocated relations with similar signs are considered to infer an undirected signed edge between two women (see \cite{doreian_multiple_2008} and how the author inferred a directed signed graph in their Fig. 5 from the data for house B.).



\bibliographystyle{plain}
\bibliography{refs}

\end{document}